\documentclass[aps,pra,groupedaddress,twocolumn,nofootinbib]{revtex4-2}
\usepackage{amsmath,amssymb,amsthm}
\usepackage{bm}
\usepackage{braket}
\usepackage{mathtools}
\usepackage{algpseudocode}
\usepackage[dvipdfmx]{graphicx}
\allowdisplaybreaks[2]

\newtheorem{corollary}{Corollary}
\newtheorem{lemma}{Lemma}
\newtheorem{proposition}{Proposition}
\newtheorem{theorem}{Theorem}
\theoremstyle{definition}
\newtheorem{definition}{Definition}

\newcommand{\A}{\mathcal{A}}
\newcommand{\bures}[2]{\mathcal{A}\left(#1,#2\right)}
\newcommand{\cov}{\theta_\mathrm{cover}}
\newcommand{\E}{\mathcal{E}}
\newcommand{\F}[2]{F\left(#1,#2\right)}
\newcommand{\Ketbra}[2]{\Ket{#1}\!\!\Bra{#2}}
\newcommand{\ketbra}[2]{\ket{#1}\!\!\bra{#2}}
\newcommand{\M}{\mathcal{M}}
\newcommand{\perr}{p_\mathrm{err}}
\newcommand{\qcdp}[2]{\mathrm{QCDP}((#1),\allowbreak(#2))}
\newcommand{\qcgdp}[3]{\mathrm{QCGDP}((#1),\allowbreak(#2),(#3))}
\newcommand{\Tr}{\mathrm{Tr}}

\DeclareMathOperator{\argmax}{argmax}

\usepackage[normalem]{ulem}
\usepackage{color}

\begin{document}

\title{Lower bounds on the error probability of quantum channel discrimination by the Bures angle and the trace distance}

\author{Ryo Ito}
\email[]{ito.r.al@m.titech.ac.jp}
\author{Ryuhei Mori}
\email[]{mori@c.titech.ac.jp}
\affiliation{School of Computing, Tokyo Institute of Technology, Japan}

\date{\today}

\begin{abstract}
Quantum channel discrimination is a fundamental problem in quantum information science.
In this study, we consider general quantum channel discrimination problems, 
and derive the lower bounds of the error probability.
Our lower bounds are based on the triangle inequalities of the Bures angle 
and the trace distance.
As a consequence of the lower bound based on the Bures angle, we prove 
the optimality of Grover's search if the number of marked elements 
is fixed to some integer $\ell$.
This result generalizes Zalka's result for $\ell=1$.
We also present several numerical results in which our lower bounds based 
on the trace distance outperform recently obtained lower bounds.
\end{abstract}

\maketitle

\section{Introduction}\label{intro}
The quantum channel discrimination problem is a fundamental problem in 
quantum information science~\cite{harrow2010adaptive,duan2009perfect,kawachi2019quantum,pirandola2019fundamental,zhuang2020ultimate,PhysRevA.104.052406}.
In this study, we consider a general quantum channel discrimination problem,
and derive the lower bounds of the error probability.

From Helstrom's seminar work~\cite{helstrom1967detection} and Holevo's 
subsequent work~\cite{holevo1972analogue}, the error probability of the discrimination 
of two quantum states $\rho$ and $\sigma$ is characterized by the trace 
distance $\frac{1}{2}\|\rho-\sigma\|_1$, which can be computed efficiently.
As a consequence, the error probability of the discrimination of two quantum channels 
$\Psi$ and $\Phi$ by algorithms invoking the quantum channel once
can be characterized by the maximum of the trace distance, 
$\max_{\rho}\frac{1}{2}\|\Phi\otimes\mathrm{id}(\rho)-\Psi\otimes\mathrm{id}(\rho)\|_1 =: \frac12\|\Phi-\Psi\|_{\diamond}$,
which is called the diamond distance.
The diamond distance can be computed efficiently as well~\cite{watrous2013simpler}.
However, if we consider algorithms that use the quantum channel $n$ times, the analysis is much more complicated.
If we restrict the discrimination algorithms to be non-adaptive, i.e., the given channel is called $n$ times at once in parallel,
the error probability is characterized by the diamond distance 
$\frac{1}{2}\|\Phi^{\otimes n}-\Psi^{\otimes n}\|_\diamond$.
To compute the above diamond distance, we have to solve an 
optimization problem on exponentially high-dimensional linear space in $n$.
Hence, it is computationally hard to compute the error probability of 
non-adaptive algorithms with $n$ queries when $n$ is large.
Furthermore, for general discrimination algorithms, no simple representation 
of the error probability has been known even for the discrimination of two quantum channels.

The exact error probability for a fixed number of queries is known if 
two unitary channels are given with uniform probabilities~\cite{kawachi2019quantum}.
In this case, a non-adaptive algorithm achieves the minimum error probability. 
For general quantum channel discriminations, however, 
there exists some discrimination problem
in which some adaptive algorithm gives a strictly smaller error probability
than any non-adaptive algorithms~\cite{harrow2010adaptive,PhysRevA.104.052406}.

The well-known Grover's search problem can be regarded as
an instance of a quantum channel discrimination problem if there is 
exactly one marked element.
Grover's search algorithm solves this problem with $O(\sqrt{N})$ queries,
where $N$ is the number of elements.
The asymptotic optimality of Grover's algorithm was shown 
in~\cite{bennett1997strengths}.
Zalka showed that the error probability of Grover's algorithm 
for a fixed number of queries is optimal when there is exactly one marked element~\cite{zalka1999grover}.

Recently, lower bounds on the error probability of general quantum channel 
discrimination have been studied~\cite{pirandola2019fundamental,zhuang2020ultimate,PhysRevLett.126.200502}.
In this study, we derive several lower bounds on the error probability of adaptive algorithms for quantum channel discrimination problems.
On the lines of previous studies~\cite{bennett1997strengths,boyer1998tight,zalka1999grover}, 
the main technique we use to derive the lower bounds is the application of triangle inequalities 
for some distance measures to series of quantum states.
In this work, we use two distance measures, namely, the Bures angle and the trace distance.

In this paper, we first derive lower bounds on the error probability of two quantum channel discriminations.
Our technique using the Bures angle is a natural generalization of the known technique for unitary channel discrimination~\cite{kawachi2019quantum}. 
For the discrimination of two amplitude damping channels with damping rates 
$r_0$ and $r_1$, we obtain a simple analytic lower bound $\frac12(1-\sin (n\Delta(r_0,r_1)))$
if $n\Delta(r_0,r_1)\le\frac{\pi}2$, where $\Delta(r_0,r_1):=\arccos(\sqrt{r_0r_1}+\sqrt{(1-r_0)(1-r_1)})$.
This lower bound is better than the known lower bounds~\cite{pirandola2019fundamental, pereira2021bounds} 
for some choices of the damping rates $r_0$ and $r_1$ and the number $n$ of queries. 
For the lower bound based on the trace distance, we introduce 
``weights'' for the series of quantum states to tighten the triangle inequality.
We present several numerical results in which our lower bounds for two quantum channel 
discrimination based on the trace distance outperform recently obtained lower bounds.

We next consider a quantum channel group discrimination problem,
which is a generalization of the quantum channel discrimination problem,
and derive lower bounds of the error probability.
Grover's search problem is an instance of the quantum channel group discrimination problem if the number of marked elements is fixed to some integer $\ell$.
Our lower bound based on the Bures angle shows the optimality of Grover's algorithm in this case.
This result generalizes Zalka's result for $\ell=1$.
Our lower bound based on the trace distance yields better lower bounds for some problems considered in~\cite{zhuang2020ultimate}.

This paper is organized as follows.
Distance measures for quantum states, including the Bures angle and the 
trace distance, are introduced in Section~\ref{sec:dist}.
The quantum channel discrimination problem and the quantum channel group 
discrimination problem are introduced in Section~\ref{sec:prelim}.
The main results in this paper are presented in Section~\ref{sec:result}.
The applications of our lower bounds are described in Section~\ref{sec:application}.
The proofs of the main results are presented in Section~\ref{sec:proof}.
In Section~\ref{sec:sdp}, we show that the lower bounds derived in this 
paper can be computed efficiently using semidefinite programming (SDP)\@.
Finally, we add some conclude remarks in Section~\ref{sec:concl}.

\section{The trace distance and the Bures angle}\label{sec:dist}
In this section, we introduce the notion of distance measures for quantum 
states, and several examples of distance measures, including the trace distance and the Bures angle.
Let $\mathsf{L}(A)$ be a set of linear operators on a quantum system $A$.
Let $\mathsf{D}(A)\subseteq \mathsf{L}(A)$ be a set of density operators on a quantum system $A$.
\begin{definition}[Distance measure for quantum states]
  A function 
  \begin{equation*}
    D\colon \{(\rho_A\in\mathsf{D}(A),\sigma_A\in\mathsf{D}(A))\mid A\colon \text{a system}\}\to [0,\infty)
  \end{equation*}
  is called a distance function if it satisfies the following conditions.\footnote{In fact, we only need the triangle inequality and the monotonicity in this paper.}
  \begin{enumerate}
    \item (Identity of indiscernibles) For any $\rho_A,\sigma_A\in\mathsf{D}(A)$, $D(\rho_A,\sigma_A)=0$ if and only if $\rho_A=\sigma_A$.
    \item (Symmetry) For any $\rho_A,\sigma_A\in\mathsf{D}(A)$, $D(\rho_A,\sigma_A)=D(\sigma_A,\rho_A)$.
    \item (Triangle inequality) For any $\rho_A,\sigma_A,\tau_A\in\mathsf{D}(A)$, $D(\rho_A,\sigma_A)\le D(\rho_A,\tau_A)+D(\tau_A,\sigma_A)$.
    \item (Monotonicity) For any $\rho_A,\sigma_A\in\mathsf{D}(A)$ and any quantum channel $\Psi_{A\to B}$ from a quantum system $A$ to a quantum system $B$, $D(\Psi_{A\to B}(\rho_A),\Psi_{A\to B}(\sigma_A))\le D(\rho_A,\sigma_A)$.
  \end{enumerate}
\end{definition}
The trace distance is a well-known distance measure for quantum states.
\begin{definition}[Trace distance]
For a linear operator $L_A\in\mathsf{L}(A)$, the trace norm is defined as $\|L_A\|_1:= \Tr\sqrt{L_AL_A^\dagger}$.
For density operators $\rho_A,\sigma_A\in\mathsf{D}(A)$, 
the trace distance is defined as $D(\rho_A,\sigma_A)=\frac12\|\rho_A-\sigma_A\|_1$.
\end{definition}
The trace distance has a useful property called the \textit{joint convexity}, which is 
$D\left(\sum_i p_i \rho_A^{(i)},\sum_i p_i \sigma_A^{(i)}\right)\le \sum_i p_i D(\rho_A^{(i)},\sigma_A^{(i)})$
for any probability distribution $(p_i)_i$ and density operators $(\rho_A^{(i)})_i$ and $(\sigma_A^{(i)})_i$.

The fidelity of quantum states is defined as follows.
\begin{definition}[Fidelity]
  For density operators $\rho_A,\sigma_A\in\mathsf{D}(A)$, the fidelity is defined as
  \begin{equation*}
  F(\rho_A,\sigma_A) := \|\sqrt{\rho_A}\sqrt{\sigma_A}\|_1.
  \end{equation*}
\end{definition}
The Bures angle is defined as follows.
\begin{definition}[Bures angle]
  For density operators $\rho_A,\sigma_A\in\mathsf{D}(A)$, the Bures angle is defined as
  \begin{equation*}
    \mathcal{A}(\rho_A, \sigma_A) := \arccos F(\rho_A,\sigma_A).
  \end{equation*}
  The Bures angle satisfies the conditions on the distance measure for quantum states~\cite{nielsen_chuang_2010}.
\end{definition}

There are several distance measures induced by the fidelity.
Indeed, the Bures distance $\mathcal{B}(\rho_A,\sigma_A)$ and 
the sine distance $\mathcal{C}(\rho_A,\sigma_A)$, defined as
\begin{align*}
  \mathcal{B}(\rho_A,\sigma_A) &:= \sqrt{2-2F(\rho_A,\sigma_A)}\\
  \mathcal{C}(\rho_A,\sigma_A) &:= \sqrt{1-F(\rho_A,\sigma_A)^2},
\end{align*}
satisfy the conditions on the distance measure for quantum states.
These distance measures can be naturally represented in terms of the Bures angle.
Indeed, for $\theta = \mathcal{A}(\rho_A,\sigma_A)$, it holds $\mathcal{B}(\rho_A,\sigma_A) = 2\sin(\theta/2)$ and 
$\mathcal{C}(\rho_A,\sigma_A)=\sin(\theta)$.
The triangle inequalities for the Bures distance and the sine distance are 
obtained from the triangle inequality for the Bures angle (see Appendix~\ref{apx:triangle}).
This fact means that the triangle inequality for the Bures angle gives 
better bounds than the triangle inequalities for the Bures distance and the sine distance.
Hence, in this study, we use the Bures angle, and do not use the Bures distance 
or the sine distance to derive inequalities.

Although the fidelity-based distances $\mathcal{A}$, $\mathcal{B}$, and 
$\mathcal{C}$ are not jointly convex, the fidelity has a useful property called the \text{joint concavity},
which is $F\left(\sum_i p_i \rho_A^{(i)},\sum_i p_i \sigma_A^{(i)}\right)\ge \sum_i p_i F(\rho_A^{(i)},\sigma_A^{(i)})$
for any probability distribution $(p_i)_i$ and density operators $(\rho_A^{(i)})_i$ and $(\sigma_A^{(i)})_i$.

\section{Quantum channel discrimination problems}\label{sec:prelim}
Let $A$ and $B$ be finite-dimensional quantum systems.
The quantum channel discrimination problem is defined as follows.
\begin{definition}[Quantum channel discrimination problem]
  Let $(\mathcal{O}_{A\to B}^\xi)_{\xi\in\Xi}$ be a finite family of quantum channels and
  $(p_\xi)_{\xi\in\Xi}$ be a probability distribution on the quantum channels.
  Assume that one of the quantum channels $\mathcal{O}_{A\to B}^\xi$ is given 
  as an oracle with probability $p_\xi$.
  The quantum channel discrimination problem
  is the problem of determining the index of the given oracle $\xi\in\Xi$ by a quantum algorithm accessing the oracle multiple times.
  This problem is denoted by $\qcdp{p_{\xi}}{\mathcal{O}_{A\to B}^{\xi}}$.
\end{definition}

We also consider the quantum channel group discrimination problem, 
which is a natural generalization of the quantum channel discrimination problem. 
\begin{definition}[Quantum channel group discrimination problem]
  Let $(\mathcal{O}_{A\to B}^\xi)_{\xi\in\Xi}$ be a finite family of quantum channels and
  $(p_\xi)_{\xi\in\Xi}$ be a probability distribution on the quantum channels.
  Let $(C_\eta\ \subseteq\ \Xi)_{\eta\in H}$ be a finite family of subsets of $\Xi$.
  Assume that one of the quantum channels $\mathcal{O}_{A\to B}^\xi$
  is given in the same way as $\qcdp{p_\xi}{\mathcal{O}_{A\to B}^\xi}$.
  The quantum channel group discrimination problem
  is the problem of finding one of the indexes of a subset $\eta\in H$ 
  that satisfies $\xi\in C_\eta$. 
  This problem is denoted by $\qcgdp{p_\xi}{\mathcal{O}_{A\to B}^\xi}{C_\eta}$.
\end{definition}
For example, let $\Xi=\{1,2,3,4\}$, $H=\{A,B,C\}$, $C_A=\{1,2\}, C_B=\{2,3\}$ and $C_C=\{1,3\}$.
When $\xi = 1$, i.e., the given oracle is $\mathcal{O}_{A\to B}^1$, the output is regarded as a correct answer if it is either of $A$ or $C$.
When $\xi = 4$, any output is an incorrect answer.
The quantum channel discrimination problem $\qcdp{p_{\xi}}{\mathcal{O}_{A\to B}^{\xi}}$ can be regarded as 
the quantum channel group discrimination problem $\qcgdp{p_\xi}{\mathcal{O}_{A\to B}^\xi}{C_\eta}$
where $H=\Xi$ and $C_\eta = \{\eta\}$ for all $\eta\in H$.

As shown in Section~\ref{sec:grover}, Grover's search problem can be regarded as an instance of quantum channel group discrimination problem.
Another important example of the quantum channel group discrimnation problem is the problem of evaluation of Boolean functions~\cite{beals2001quantum,ambainis2016quantum}.

We next consider algorithms that solve the quantum channel group discrimination problem.
The general algorithm for the quantum channel group discrimination problem is formulated as follows.

\begin{definition}[Adaptive discrimination algorithm]
  Let $R$ be a working system.
  Let $(\Phi_{BR\to AR}^i)_{i=1}^n$ be a family of quantum channels,
  and $(M_{BR}^\eta)_{\eta\in H}$ be a positive operator-valued measure (POVM)\@. 
  Then a pair $((\Phi_{BR\to AR}^i),(M_{BR}^\eta))$
  represents the following adaptive algorithm for $\qcgdp{p_\xi}{\mathcal{O}_{A\to B}^\xi}{C_\eta}$.
  \begin{algorithmic}[1]
    \State Set the initial state $\ketbra{0}{0}_{BR}$ in a quantum computer.
    \For{$i=1,2,\dotsc,n$}
      \State Apply the quantum channel $\Phi_{BR\to AR}^i$.
      \State Apply the given oracle $\mathcal{O}_{A\to B}^\xi$.
    \EndFor
    \State Apply a measurement $(M_{BR}^\eta)_{\eta\in H}$ to the state
    and output the observed value $\eta\in H$.
  \end{algorithmic}
\end{definition}

An algorithm succeeds if and only if its output $\eta\in H$ satisfies 
$\xi\in C_\eta$, where $\xi\in\Xi$ is the index of the given oracle $\mathcal{O}^\xi_{A\to B}$.
The minimum error probability of a QCGDP is then defined as follows.

\begin{definition}[The minimum error probability of a QCGDP]
  Let $((\Phi_{BR\to AR}^i),(M_{BR}^\eta))$ 
  be an adaptive algorithm for $\qcgdp{p_\xi}{\mathcal{O}_{A\to B}^\xi}{C_\eta}$ with $n$ queries. 
  For each $\xi\in\Xi$, a density operator $\rho_{BR}^\xi$ is defined as
  \begin{align*}
    \rho_{BR}^\xi&\coloneqq \mathcal{O}_{A\to B}^\xi\circ \Phi^n_{BR\to AR}\circ \mathcal{O}_{A\to B}^\xi\circ \Phi^{n-1}_{BR\to AR}\\*
    &\quad\circ\dotsm \circ\mathcal{O}_{A\to B}^\xi \circ \Phi^1_{BR\to AR}(\ketbra{0}{0}_{BR}).
  \end{align*}
  The minimum error probability of $\qcgdp{p_\xi}{\mathcal{O}_{A\to B}^\xi}{C_\eta}$
  with $n$ adaptive queries is defined as
  \begin{align*}
    \perr(n) \coloneqq 
    \min_{(\Phi_{BR\to AR}^i),(M_{BR}^\eta)}
    \sum_{\xi\in\Xi}p_\xi\sum_{\eta\colon\xi\notin C_\eta}
    \Tr\left(M_{BR}^\eta\rho_{BR}^{\xi}\right).
  \end{align*}
  The minimum error probability of $\qcdp{p_\xi}{\mathcal{O}_{A\to B}^\xi}$
  is defined in the same way.
\end{definition}
An algorithm that achieves the minimum error probability exists for any QCGDP
because the set of algorithms is compact and the error probability is continuous.
The main goal of this study is to derive lower bounds of the minimum error 
probability $\perr(n)$.

The error probability $\perr(0)$ without calling the oracle can be 
evaluated as follows:
\begin{align*}
  \perr(0)
  &=\min_{(M_{BR}^\eta)}
  \sum_{\xi\in\Xi}p_\xi\sum_{\eta\colon\xi\notin C_\eta}
  \Tr\left(M_{BR}^\eta\ketbra{0}{0}_{BR}\right)\\
  &=\min_{(M_{BR}^\eta)}
  \sum_{\eta\in H}\Tr\left(M_{BR}^\eta\ketbra{0}{0}_{BR}\right)
  \sum_{\xi\notin C_\eta}p_\xi\\
  &=\min_{\eta\in H}\left(\sum_{\xi\notin C_\eta}p_\xi\right). 
\end{align*}

Once the quantum channels $(\Phi_{BR\to AR}^i)_{i=1}^n$ in 
a discrimination algorithm are fixed, the problem is reduced to 
the quantum state discrimination problem.
The optimal measurement and the error probability of discrimination 
of two quantum states were known by Helstrom~\cite{helstrom1967detection} and Holevo~\cite{holevo1972analogue}.

\begin{proposition}[Holevo--Helstrom theorem~\cite{helstrom1967detection,holevo1972analogue,watrous2018theory}]\label{holevoHelstrom}
  Let $\rho_A^0,\rho_A^1\in\mathsf{D}(A)$ be density operators.
  Let $p_0,p_1\in[0,1]$ be non-negative real numbers
  that satisfy $p_0+p_1=1$.
  For any POVM $(M_A^0,M_A^1)$, the following holds:
  \[
    \sum_{\xi\in\{0,1\}}p_\xi\Tr\left(M_A^\xi\rho_A^\xi\right)
    \leq\frac{1}{2}\left(1+\left\|p_0\rho_A^0-p_1\rho_A^1\right\|_1\right).
  \]
  Moreover, a POVM $(M_A^0,M_A^1)$ exists that satisfies the equality.
\end{proposition}

From the Holevo--Helstrom theorem, the error probability of discrimination
of two quantum channels with a single query is given by the following proposition.

\begin{proposition}[\cite{watrous2018theory}]
  Let $\perr(1)$ be the minimum error probability for 
  $\qcdp{p_0,p_1}{\mathcal{O}_{A\to B}^0,\mathcal{O}_{A\to B}^1}$ with one query.
  The following then holds:
  \begin{align*}
    \perr(1)=\frac{1}{2}\Biggl(1
    &-\max_{\rho_{AR}\in\mathsf{D}(AR)}\Bigl\|p_0(\mathcal{O}_{A\to B}^0\otimes\mathrm{id}_R)(\rho_{AR})\\
    &-p_1(\mathcal{O}_{A\to B}^1\otimes\mathrm{id}_R)(\rho_{AR})\Bigr\|_1\Biggr),
  \end{align*}
  where $\mathrm{id}_R$ denotes the identity map on $\mathsf{L}(R)$.
\end{proposition}
If $\dim(R)\ge\dim(A)$, the maximum of the trace norm is called the diamond norm.

\begin{definition}[Diamond norm]
  Let $\Psi_{A\to B}$ be a linear map from $\mathsf{L}(A)$ to $\mathsf{L}(B)$.
  The diamond norm of $\Psi_{A\to B}$ is then defined as
  \begin{equation*}
    \|\Psi_{A\to B}\|_\diamond\coloneqq\max_{L_{AR}\in\mathsf{L}(AR)}
    \left\|(\Psi_{A\to B}\otimes\mathsf{id}_R)(L_{AR})\right\|_1,
  \end{equation*}
  for a quantum system $R$ whose dimension is equal to $\dim(A)$.
\end{definition}
From the convexity of the trace norm, we can assume that $L_{AR}$ in the maximization problem is a rank-1 matrix.
Especially, in this paper, we consider the diamond norm of Hermitian-preserving map.
In that case, we can assume that $L_{AR}$ is a density operator of pure state~\cite{watrous2018theory}.

It is known that the following holds for any
Hermitian preserving map $\Lambda_{A\to B}$.
\begin{equation*}
  \|\Lambda_{A\to B}\|_\diamond=\max_{\ket{\phi}_{AR}\in\mathsf{S}(AR)}
  \left\|(\Lambda_{A\to B}\otimes\mathsf{id}_R)(\ketbra{\phi}{\phi}_{AR})\right\|_1,
\end{equation*}
where $\mathsf{S}(A)$ denotes a set of state vectors on a quantum system $A$~\cite{watrous2018theory}.

An algorithm is said to be non-adaptive if it calls the oracle in parallel.
In general, the minimum error probability of non-adaptive algorithms with $n$ queries for discriminating two quantum channels is expressed as
\begin{align*}
  \frac{1}{2}\left(1-
  \left\|p_0\left(\mathcal{O}_{A\to B}^0\right)^{\otimes n}-p_1\left(\mathcal{O}_{A\to B}^1\right)^{\otimes n}\right\|_\diamond\right),
\end{align*}
if $\dim(R)\ge \dim(A)^n$.
Although the diamond norm $\|\Psi_{A\to B}\|_\diamond$ can be evaluated in 
polynomial time with respect to the dimensions of $A$ and $B$ via SDP~\cite{watrous2013simpler},
the dimension of $A^{\otimes n}$ is exponentially large, so that SDP 
does not give an efficient algorithm when $n$ is large.
For general adaptive algorithms, no simple formula expressing the minimum 
error probability is known even for the discrimination of two quantum channels.

\section{Main results}\label{sec:result}
In this paper, we present four theorems on the lower bounds 
on the minimum error probability $\perr(n)$.
Let $E$ be a finite-dimensional quantum system.
A Stinespring representation $O_{A\to BE}$ of a quantum channel 
$\mathcal{O}_{A\to B}$ is a linear isometry from $A$ to $B\otimes E$ satisfying
\begin{equation*}
  \mathcal{O}_{A\to B}(\rho_A)
  =\Tr_E\left(O_{A\to BE}\,\rho_A\, O_{A\to BE}^\dagger\right).
\end{equation*}

We first show a lower bound of the error probability 
for the discrimination of two quantum channels by using the Bures angle.
\begin{theorem}\label{boundForQCDPWithBuresAngle}
  Let $\perr(n)$ be the minimum error probability for 
  $\qcdp{p_0,p_1}{\mathcal{O}_{A\to B}^0,\mathcal{O}_{A\to B}^1}$ with $n$ adaptive queries. 
  Let $\tau_\A$ be
  \begin{align*}
    \tau_\A&\coloneqq\arccos\\
    &\min_{\ket{\phi}_{AR}\in\mathsf{S}(AR)}
    \F{\mathcal{O}_{A\to B}^0(\ketbra{\phi}{\phi}_{AR})}{\mathcal{O}_{A\to B}^1(\ketbra{\phi}{\phi}_{AR})}.
  \end{align*}
  Here, the identity map acts on the working system $R$.
  The following then holds:
  \begin{align}
    \perr(n)\geq\frac{1}{2}\left(1-\sqrt{1-4p_0p_1\cos^2(n\tau_\A)}\right),
\label{eq:A2}
  \end{align}
  under the condition $n\tau_\A\le\pi/2$.
  Furthermore, if both the quantum channels are unitary channels,
  then there exists a discrimination algorithm that achieves the lower bound even if $\dim(R)=0$.

  Furthermore, if $\dim(R)\geq\dim(A)$, the following then holds:
  \begin{equation}
    \tau_\A=\arccos\min_{\sigma_A\in\mathsf{D}(A)}\left\|
      \Tr_B\left(O_{A\to BE}^0\sigma_AO_{A\to BE}^{1\dagger}\right)
    \right\|_1,
    \label{eq:ft}
  \end{equation}
  where $O_{A\to BE}^0$ and $O_{A\to BE}^1$ are Stinespring 
  representations of $\mathcal{O}_{A\to B}^0$ and $\mathcal{O}_{A\to B}^1$, respectively.
\end{theorem}
Note that the minimum of the fidelity
\begin{align*}
  &\min_{\ket{\phi}_{AR}\in\mathsf{S}(AR)}
  \F{\mathcal{O}_{A\to B}^0(\ketbra{\phi}{\phi}_{AR})}{\mathcal{O}_{A\to B}^1(\ketbra{\phi}{\phi}_{AR})}
\end{align*}
is referred to as the stabilized process fidelity in~\cite{PhysRevA.71.062310}.

We next show a lower bound based on the trace distance 
for the discrimination of two quantum channels.
\begin{theorem}\label{boundForQCDPWithWeightedTraceDistance}
  Let $\perr(n)$ be the minimum error probability for
  $\qcdp{p_0,p_1}{\mathcal{O}_{A\to B}^0,\mathcal{O}_{A\to B}^1}$ 
  with $n$ adaptive queries.
  Let $k\in\{0,1,\dotsc,n\}$.
  Let $\alpha_0,\alpha_1\in[0,\infty)$ be non-negative real numbers
  that satisfy $p_0\alpha_0^k=p_1\alpha_1^{n-k}$.
  Let $\tau_\diamond^0$ and $\tau_\diamond^1$ be
  \begin{align*}
    \tau_\diamond^0&\coloneqq\max_{\Ket{\phi}_{AR}\in\mathsf{S}(AR)}\frac{1}{2}\left\|
      \left(\mathcal{O}_{A\to B}^0-\alpha_0\mathcal{O}_{A\to B}^1\right)(\Ketbra{\phi}{\phi}_{AR})
    \right\|_1,\\
    \tau_\diamond^1&\coloneqq\max_{\Ket{\phi}_{AR}\in\mathsf{S}(AR)}\frac{1}{2}\left\|
      \left(\alpha_1\mathcal{O}_{A\to B}^0-\mathcal{O}_{A\to B}^1\right)(\Ketbra{\phi}{\phi}_{AR})
    \right\|_1.
  \end{align*}
  The following then holds:
  \begin{align*}
    \perr(n)\geq\frac{1}{2}
      -p_0\left(\sum_{i=0}^{k-1}\alpha_0^i\right)\tau_\diamond^0
      -p_1\left(\sum_{i=0}^{n-k-1}\alpha_1^i\right)\tau_\diamond^1,
  \end{align*}
  for arbitrary $\alpha_0$, $\alpha_1$ and, $k$.

  Furthermore, if $\dim(R)\geq\dim(A)$, the following then holds~\cite{watrous2018theory}:
  \begin{align*}
    \tau_\diamond^0
    &=\frac{1}{2}\left\|\mathcal{O}_{A\to B}^0-\alpha_0\mathcal{O}_{A\to B}^1\right\|_\diamond,\\
    \tau_\diamond^1
    &=\frac{1}{2}\left\|\alpha_1\mathcal{O}_{A\to B}^0-\mathcal{O}_{A\to B}^1\right\|_\diamond.
  \end{align*}
\end{theorem}

We also show a lower bound of the error probability 
for QCGDP by using the Bures angle.
\begin{theorem}\label{boundForQCGDPWithBuresAngle}
  Let $\perr(n)$ be the minimum error probability for 
  $\qcgdp{p_\xi}{\mathcal{O}_{A\to B}^\xi}{C_\eta}$ with $n$ adaptive queries.
  For a quantum channel $\Psi_{AR\to BR}$ and $m\in\{0,1,\dotsc,n\}$, 
  $\theta_\A(\Psi_{AR\to BR})$ and $\theta_m$ are defined as
  \begin{align*}
    &\theta_\A(\Psi_{AR\to BR})\coloneqq\arccos\min_{\ket{\phi}_{AR}\in\mathsf{S}(AR)}\sum_{\xi\in\Xi}p_\xi\\
    &\qquad\qquad\F{\mathcal{O}_{A\to B}^\xi(\ketbra{\phi}{\phi}_{AR})}{\Psi_{AR\to BR}(\ketbra{\phi}{\phi}_{AR})},\\
    &\theta_m\coloneqq\arccos\sqrt{\perr(m)}.
  \end{align*}
  The following then holds:
  \begin{align*}
    \perr(n)\geq\cos^2\bm((n-m)\theta_\A(\Psi_{AR\to BR})+\theta_m\bm),
  \end{align*}
  for arbitrary $\Psi_{AR\to BR}$ and $m$
  satisfying $(n-m)\theta_\A(\Psi_{AR\to BR})+\theta_m\le\pi/2$.

  Furthermore, if $\dim(R)\geq\dim(A)$, the following then holds:
  \begin{align*}
    \theta_\A(\Psi_{A\to B}\otimes\mathrm{id}_R)
    &=\arccos\min_{\sigma_A\in\mathsf{D}(A)}\sum_{\xi\in\Xi}p_\xi\\
    &\left\|\Tr_B\left(O_{A\to BE}^\xi\sigma_{A}(V_{A\to BE})^\dagger\right)\right\|_1,
  \end{align*}
  where $O_{A\to BE}^\xi$ is a Stinespring representation of 
  $\mathcal{O}_{A\to B}^\xi$ for $\xi\in\Xi$, and $V_{A\to BE}$ is 
  a Stinespring representation of $\Psi_{A\to B}$.
\end{theorem}

Finally, we show a lower bound of the error probability for QCGDP 
by using the trace distance.
\begin{theorem}\label{boundForQCGDPWithWeightedTraceDistance}
  Let $\perr(n)$ be the minimum error probability for 
  $\qcgdp{p_\xi}{\mathcal{O}_{A\to B}^\xi}{C_\eta}$ with $n$ adaptive queries.
  Let $m$ and $k$ be integers that satisfy $0\leq m\leq k\leq n$.
  Let $\alpha_0,\alpha_1\in[0,\infty)$ be non-negative real numbers
  that satisfy $\alpha_0^{n-k}=\alpha_1^{k-m}$.
  For quantum channels $\Psi_{AR\to BR}^0$ and $\Psi_{AR\to BR}^1$,
  $\theta_\diamond^0(\Psi_{AR\to BR}^0)$ and
  $\theta_\diamond^1(\Psi_{AR\to BR}^1)$ are defined as
  \begin{align*}
    &\theta_\diamond^0(\Psi_{AR\to BR}^0)
    \coloneqq\max_{\ket{\phi}_{AR}\in\mathsf{S}(AR)}\sum_{\xi\in\Xi}p_\xi\\
    &\qquad\qquad\frac{1}{2}\left\|
      \left(\mathcal{O}_{A\to B}^\xi-\alpha_0\Psi_{AR\to BR}^0\right)(\ketbra{\phi}{\phi}_{AR})
    \right\|_1,\\
    &\theta_\diamond^1(\Psi_{AR\to BR}^1)
    \coloneqq\max_{\ket{\phi}_{AR}\in\mathsf{S}(AR)}\sum_{\xi\in\Xi}p_\xi\\
    &\qquad\qquad\frac{1}{2}\left\|
      \left(\alpha_1\mathcal{O}_{A\to B}^\xi-\Psi_{AR\to BR}^1\right)(\ketbra{\phi}{\phi}_{AR})
    \right\|_1,
  \end{align*}
  respectively. The following then holds:
  \begin{align*}
    \perr(n)\geq\perr(m)
    &-\left(\sum_{i=0}^{n-k-1}\alpha_0^i\right)\theta_\diamond^0(\Psi_{AR\to BR}^0)\\
    &-\left(\sum_{i=0}^{k-m-1}\alpha_1^i\right)\theta_\diamond^1(\Psi_{AR\to BR}^1),
  \end{align*}
  for arbitrary $\Psi_{AR\to BR}^0$, $\Psi_{AR\to BR}^1$, $m$, $\alpha_0$, $\alpha_1$, and $k$.
\end{theorem}

The proofs of Theorems~\ref{boundForQCDPWithBuresAngle},
\ref{boundForQCDPWithWeightedTraceDistance},~\ref{boundForQCGDPWithBuresAngle},
and~\ref{boundForQCGDPWithWeightedTraceDistance} are presented in Section~\ref{sec:proof}.
In Section~\ref{sec:sdp}, we show that all optimization problems on 
quantum states $\ket{\phi}_{AR}\in\mathsf{S}(AR)$ for computing the lower bounds in Theorems~\ref{boundForQCDPWithBuresAngle},
\ref{boundForQCDPWithWeightedTraceDistance},~\ref{boundForQCGDPWithBuresAngle},
and~\ref{boundForQCGDPWithWeightedTraceDistance} can be written as SDP when 
$\dim(R)\geq\dim(A)$, so that the lower bounds can be evaluated efficiently.
In this case, the optimization of the lower bounds in Theorems~\ref{boundForQCGDPWithBuresAngle} and~\ref{boundForQCGDPWithWeightedTraceDistance}
with respect to the quantum channel $\Psi_{AR\to BR}$ with a constraint $\Psi_{AR\to BR}=\Psi_{A\to B}\otimes\mathrm{id}_R$ for some quantum channel $\Psi_{A\to B}$
can be solved efficiently as well.

\section{Applications}\label{sec:application}
\subsection{Application of Theorems~\ref{boundForQCDPWithBuresAngle} and~\ref{boundForQCDPWithWeightedTraceDistance}: Discrimination of two amplitude damping channels}\label{sec:qadc_two}
In this section, we derive lower bounds on the error probability of the discrimination of 
two amplitude damping channels by using Theorems~\ref{boundForQCDPWithBuresAngle} and~\ref{boundForQCDPWithWeightedTraceDistance}.
\begin{definition}
  The amplitude damping channel $\E_A^r$ with damping rate 
  $r\in[0,1]$ is defined by a Kraus representation
  \begin{gather*}
    \E_A^r(\rho_A)\coloneqq K_A^{0,r}\rho_A(K_A^{0,r})^\dagger
    +K_A^{1,r}\rho_A(K_A^{1,r})^\dagger,\\
    K_A^{0,r}\coloneqq\Ketbra{0}{0}_A+\sqrt{1-r}\Ketbra{1}{1}_A,\,
    K_A^{1,r}\coloneqq\sqrt{r}\Ketbra{0}{1}_A.
  \end{gather*}
\end{definition}
Let $E$ be a two-dimensional quantum system.
Let $U_{A\to AE}^r\coloneqq K_A^{0,r}\otimes\ket{0}_E+K_A^{1,r}\otimes\ket{1}_E$,
then $U_{A\to AE}^r$ is a Stinespring representation of $\E_A^r$.
For $\qcdp{p_0,p_1}{\E_A^{r_0},\E_A^{r_1}}$, we obtain the following 
bound of $\tau_\A$ in Theorem~\ref{boundForQCDPWithBuresAngle}
with a two-dimensional working system $R$:
\begin{align*}
  \tau_\A
  &=\arccos\min_{\sigma_A\in\mathsf{D}(A)}\left\|
    \Tr_A\bm(U_{A\to AE}^{r_1}\sigma_A(U_{A\to AE}^{r_0})^\dagger\bm)
  \right\|_1\\
  &\leq\arccos\min_{\sigma_A\in\mathsf{D}(A)}\left|
    \Tr_{AE}\bm(U_{A\to AE}^{r_1}\sigma_A(U_{A\to AE}^{r_0})^\dagger\bm)
  \right|\\
  &=\arccos\min_{\sigma_A\in\mathsf{D}(A)}\Bigl|\Braket{0|\sigma_A|0}_A\\
  &\qquad\qquad
    +\left(\sqrt{r_0r_1}+\sqrt{(1-r_0)(1-r_1)}\right)\Braket{1|\sigma_A|1}_A
  \Bigr|\\
  &=\arccos\bm{\left(}\sqrt{r_0r_1}+\sqrt{(1-r_0)(1-r_1)}\bm{\right)}.
\end{align*}
The inequality follows from the monotonicity of the trace norm 
(i.e., $\|\cdot\|_1\geq|\Tr(\cdot)|$).
Note that $\sigma_A=\Ketbra{1}{1}_A$ satisfies the equality.
Let the Bhattacharyya angle of $r_0$ and $r_1$ be 
$\Delta(r_0,r_1):=\arccos\bm(\sqrt{r_0r_1}+\sqrt{(1-r_0)(1-r_1)}\bm)$.
We then obtain
\begin{align*}
  \perr(n)\geq\frac{1}{2}
  \left(1-\sqrt{1-4p_0p_1\cos^2\bm(n\Delta(r_0,r_1)\bm)}\right)
\end{align*}
from Theorem~\ref{boundForQCDPWithBuresAngle}
under the condition $n\Delta(r_0,r_1)\in[0,\pi/2]$.
\begin{figure}[t]
  \centering
  \includegraphics[width=\columnwidth]{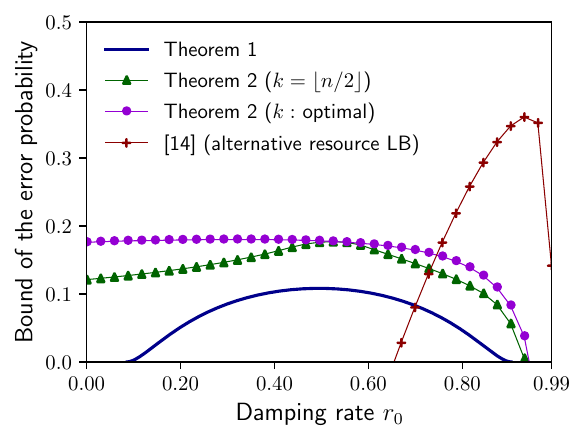}
  \caption{
    The lower bounds of two amplitude damping channel discrimination
    ($p_0=p_1=1/2, r_1=r_0+0.01,n=90$);
    continuous line, the bound from Theorem~\ref{boundForQCDPWithBuresAngle};
    $\blacktriangle$, the bound from Theorem~\ref{boundForQCDPWithWeightedTraceDistance}
    where $k=\lfloor n/2\rfloor$;
    $\bullet$, the bound from Theorem~\ref{boundForQCDPWithWeightedTraceDistance}
    where $k$ is optimal;
    $\bm{+},$ the bound from~\cite{pereira2021bounds}.
  }
  \label{fig:two_qadc_r0}
\end{figure}
\begin{figure}[t]
  \centering
  \includegraphics[width=\columnwidth]{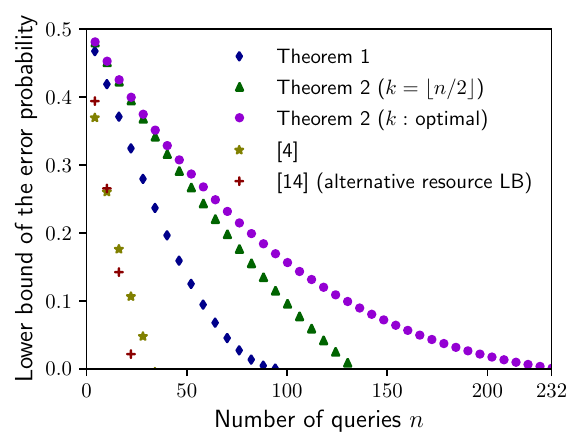}
  \caption{
    The lower bounds of two amplitude damping channel discrimination
    ($p_0=p_1=1/2, r_0=0.10,r_1=0.11$);
    $\blacklozenge$, the bound from Theorem~\ref{boundForQCDPWithBuresAngle};
    $\bigstar$, the bound from~\cite{pirandola2019fundamental};
    the other marks are the same as Fig.~\ref{fig:two_qadc_r0}.
  }
  \label{fig:two_qadc_n}
\end{figure}

Next, we consider the lower bound obtained from
Theorem~\ref{boundForQCDPWithWeightedTraceDistance}.
Let $k\in\{0,1,\dotsc,n\}$.
Let $\alpha_0,\alpha_1\in[0,\infty)$ be non-negative real numbers
that satisfy $p_0\alpha_0^k=p_1\alpha_1^{n-k}$. We then obtain
\begin{align*}
  \perr(n)\geq\frac{1}{2}\Biggl[1
    &-p_0\left(\sum_{i=0}^{k-1}\alpha_0^i\right)
    \left\|\E_A^{r_0}-\alpha_0\E_A^{r_1}\right\|_\diamond\\
    &-p_1\left(\sum_{i=0}^{n-k-1}\alpha_1^i\right)
    \left\|\alpha_1\E_A^{r_0}-\E_A^{r_1}\right\|_\diamond
  \Biggr].
\end{align*}

The numerical results for two amplitude damping channel discrimination
are shown in Figs.~\ref{fig:two_qadc_r0} and~\ref{fig:two_qadc_n}.
In both the figures, the distribution of the channels is uniform, 
i.e., $p_0=p_1=1/2$.
In Fig.~\ref{fig:two_qadc_r0}, the damping rates of the channels 
are $r_0$ and $r_1=r_0+0.01$. The number of queries is $n=90$.
In Fig.~\ref{fig:two_qadc_n}, the damping rates are $r_0=0.10$ 
and $r_1=0.11$. Compared with the lower bounds obtained by
Pirandola et al.~\cite{pirandola2019fundamental}
and Pereira and Pirandola~\cite{pereira2021bounds} (alternative resource lower bound, which is specific to the amplitude damping channel),
the lower bound given by Theorems~\ref{boundForQCDPWithBuresAngle}
and~\ref{boundForQCDPWithWeightedTraceDistance}
are better for some $r_0$ in Fig.~\ref{fig:two_qadc_r0},
and better for all $n$ in Fig.~\ref{fig:two_qadc_n}.
In Fig.~\ref{fig:two_qadc_r0}, the lower bound from~\cite{pirandola2019fundamental} is identically $0$.
Although Theorem~\ref{boundForQCDPWithBuresAngle} gives the analytic lower bound, the lower bound from Theorem~\ref{boundForQCDPWithWeightedTraceDistance}
gives better lower bound in this case.

In the numerical results, we optimized $\alpha_0$ for given $k$ by the 
golden-section search because the function empirically has a single maximal point.
We optimized $k$ in an exhaustive way under the assumption that the optimum 
$k^*$ is non-decreasing with respect to $r_0$ and $n$.

\subsection{Application of Theorem~\ref{boundForQCGDPWithBuresAngle}: Grover's search problem}\label{sec:grover}
In this section, we prove the optimality of Grover's algorithm by using Theorem~\ref{boundForQCGDPWithBuresAngle}.
Let $A$ be an $|H|$-dimensional quantum system, and 
$(\ket\eta_A)_{\eta\in H}$ be its computational basis.
For some fixed $\ell\in\{1,2,\dotsc,|H|/2\}$, 
$\Xi\coloneqq\{\xi\ \subseteq\  H\mid |\xi|=\ell\}$.
Let $(O_A^\xi)_{\xi\in\Xi}$ be a family of unitary operators
defined as $O_A^\xi \coloneqq I_A-2\sum_{u\in \xi}\Ketbra{u}{u}_A$,
where $I_A$ is the identity operator in $A$,
Let $(\mathcal{O}_A^\xi)_{\xi\in\Xi}$ be a family of quantum channels defined as
\begin{align*}
  \mathcal{O}_A^\xi(\rho_A) \coloneqq O_A^\xi \rho_A O_A^{\xi\dagger}.
\end{align*}
Let $(C_\eta)_{\eta\in H}$ be a family of subsets of $\Xi$ defined as
$C_\eta \coloneqq \{\xi\in\Xi \mid \eta\in\xi\}$.
The problem $\qcgdp{1/|\Xi|}{\mathcal{O}_A^\xi}{C_\eta}$ is then called Grover's search problem.
The following holds:
\begin{align*}
  \frac{1}{|\Xi|}\sum_{\xi\in\Xi}O_A^\xi
  &=\binom{|H|}{\ell}^{-1}
  \left[\binom{|H|-1}{\ell}-\binom{|H|-1}{\ell-1}\right]I_A\\
  &=\left(1-\frac{2\ell}{|H|}\right)I_A.
\end{align*}
Hence, we obtain
\begin{align*}
  &\theta_\A(\mathrm{id}_A)\\
  &=\arccos\min_{\ket{\phi}_{AR}\in\mathsf{S}(AR)}\sum_{\xi\in\Xi}\frac{1}{|\Xi|}
  \F{\mathcal{O}_A^\xi(\ketbra{\phi}{\phi}_{AR})}{\ketbra{\phi}{\phi}_{AR}}\\
  &=\arccos\min_{\ket{\phi}_{AR}\in\mathsf{S}(AR)}\sum_{\xi\in\Xi}\frac{1}{|\Xi|}
  \left|\Braket{\phi|O_A^\xi|\phi}_{AR}\right|\\
  &\leq\arccos\min_{\ket{\phi}_{AR}\in\mathsf{S}(AR)}
  \left|\sum_{\xi\in\Xi}\frac{1}{|\Xi|}\Braket{\phi|O_A^\xi|\phi}_{AR}\right|\\
  &=\arccos\left(1-\frac{2\ell}{|H|}\right)\\
  &=2\arcsin\sqrt{\frac{\ell}{|H|}}.
\end{align*}
Moreover, we obtain
\begin{align*}
  \theta_0
  &=\arccos\sqrt{1-
    \max_{\eta\in H}\left(\sum_{\xi\in C_\eta}\frac{1}{|\Xi|}\right)}\\
  &=\arcsin\sqrt{
    \max_{\eta\in H}\left(\sum_{\xi\in C_\eta}\frac{1}{|\Xi|}\right)}\\
  &=\arcsin\sqrt{\binom{|H|-1}{\ell-1}\binom{|H|}{\ell}^{-1}}\\
  &=\arcsin\sqrt{\frac{\ell}{|H|}}.
\end{align*}
Therefore, we obtain the lower bound
\begin{align*}
  \perr(n)\geq\cos^2\left((2n+1)\arcsin\sqrt{\frac{\ell}{|H|}}\right),
\end{align*}
from Theorem~\ref{boundForQCGDPWithBuresAngle} under 
the condition $(2n+1)\allowbreak\arcsin(\sqrt{\ell/|H|})\in[0,\pi/2]$.
This lower bound is achieved by Grover's algorithm~\cite{grover1996fast,brassard2002quantum} .
The lower bound for $\ell=1$ was obtained by 
Zalka~\cite{zalka1999grover}.\footnote{Zalka wrote: ``It seems very plausible that the proof can be 
extended to oracles with any known number of marked elements.'' 
However, no formal proof is known to the best of our knowledge.}

\subsection{Application of Theorem~\ref{boundForQCGDPWithWeightedTraceDistance}: Channel position finding for amplitude damping channels}\label{sec:qadc_multi}
In this section, we consider the channel position finding problem introduced in~\cite{zhuang2020ultimate} for amplitude damping channels to demonstrate 
the lower bound given by Theorem~\ref{boundForQCGDPWithWeightedTraceDistance}.
Let $\ell\geq 2$ be an integer.
Let $A_1,\dotsc,A_\ell$ be two-dimensional quantum systems.
Let $A'$ be the composite system of $A_1,\dotsc,A_\ell$,
i.e., $A'=\bigotimes_{i=1}^\ell A_i$.
Let $\Xi\coloneqq\{1,2,\dotsc,\ell\}$, and $r_0,r_1\in[0,1]$ .
For each $\xi\in\Xi$, a quantum channel $\mathcal{O}_{A'}^\xi$ is defined as
\[
  \mathcal{O}_{A'}^\xi\coloneqq 
  \left(\bigotimes_{i=1}^{\xi-1}\E_{A_i}^{r_1}\right)
  \otimes\E_{A_\xi}^{r_0}\otimes
  \left(\bigotimes_{i=\xi+1}^{\ell}\E_{A_i}^{r_1}\right).
\]
In these quantum channels, $\mathcal{E}_A^{r_1}$ is applied for all but one subsystem.
For the one exceptional system, $\mathcal{E}_A^{r_0}$ is applied.
Our task is to find the place where $\mathcal{E}_A^{r_0}$ is applied.
This problem is formulated as the channel discrimination problem $\qcdp{1/\ell}{\mathcal{O}_{A'}^\xi}$.
We consider the error probability of discrimination algorithms using a working system $R'=R^{\otimes\ell}$ where $R$ is a two-dimensional working system.
Let $k\in\{0,1,\dotsc,n\}$.
Let $\alpha_0,\alpha_1\in[0,\infty)$ be non-negative real numbers
that satisfy $\alpha_0^{n-k}=\alpha_1^k$. We then obtain
\begin{align*}
  &\left\|
    \mathcal{O}_{A'}^\xi-\alpha_0\bigotimes_{i=1}^\ell\E_{A_i}^{r_1}
  \right\|_\diamond\\
  &=\left\|
    \left(\bigotimes_{i=1}^{\xi-1}\E_{A_i}^{r_1}\right)
    \otimes\left(\E_{A_\xi}^{r_0}-\alpha_0\E_{A_\xi}^{r_1}\right)\otimes
    \left(\bigotimes_{i=\xi+1}^{\ell}\E_{A_i}^{r_1}\right)
  \right\|_\diamond\\
  &=\left(\prod_{i=1}^{\xi-1}\Bigl\|\E_{A_i}^{r_1}\Bigr\|_\diamond\right)
    \Bigl\|\E_{A_\xi}^{r_0}-\alpha_0\E_{A_\xi}^{r_1}\Bigr\|_\diamond
    \left(\prod_{i=\xi+1}^{\ell}\Bigl\|\E_{A_i}^{r_1}\Bigr\|_\diamond\right)\\
  &=\left\|\E_{A}^{r_0}-\alpha_0\E_{A}^{r_1}\right\|_\diamond.
\end{align*}
The second equality follows from the multiplicativity of 
the diamond norm~\cite{watrous2018theory}.

Hence, $\theta_\diamond^0\left(\bigotimes_{i=1}^\ell\E_{A_i}^{r_1}\right)$
in Theorem~\ref{boundForQCGDPWithWeightedTraceDistance} is bounded as
\begin{align*}
  &\theta_\diamond^0\left(\bigotimes_{i=1}^\ell\E_{A_i}^{r_1}\right)\\
  &=\max_{\ket{\phi}_{A'R'}\in\mathsf{S}(A'R')}\sum_{\xi\in\Xi}\frac{1}{2\ell}\left\|
    \left(\mathcal{O}_{A'}^\xi-\alpha_0\bigotimes_{i=1}^\ell\E_{A_i}^{r_1}\right)(\ketbra{\phi}{\phi}_{A'R'})
  \right\|_1\\
  &\leq\sum_{\xi\in\Xi}\frac{1}{2\ell}\left\|
    \mathcal{O}_{A'}^\xi-\alpha_0\bigotimes_{i=1}^\ell\E_{A_i}^{r_1}
  \right\|_\diamond\\
  &=\frac{1}{2}\left\|\E_A^{r_0}-\alpha_0\E_A^{r_1}\right\|_\diamond.
\end{align*}
Note that 
\[
  \ket{\phi}_{A'R'}=\left(
    \underset{\ket{\phi}_{AR}\in\mathsf{S}(AR)}{\argmax}\left\|
      (\E_A^{r_0}-\alpha_0\E_A^{r_1})(\ketbra{\phi}{\phi}_{AR})
    \right\|_1
  \right)^{\otimes\ell}
\]
satisfies the equality. Similarly, we obtain
\begin{align*}
  \theta_\diamond^1\left(\bigotimes_{i=1}^\ell\E_{A_i}^{r_1}\right)
  =\frac{1}{2}\left\|\alpha_1\E_A^{r_0}-\E_A^{r_1}\right\|_\diamond.
\end{align*}
From $\perr(0)=1-1/\ell$, we obtain
\begin{align*}
  \perr(n)\geq 1-\frac{1}{\ell}
  &-\left(\sum_{i=0}^{n-k-1}\frac{\alpha_0^i}{2}\right)
  \left\|\E_A^{r_0}-\alpha_0\E_A^{r_1}\right\|_\diamond\\
  &-\left(\sum_{i=0}^{k-1}\frac{\alpha_1^i}{2}\right)
  \left\|\alpha_1\E_A^{r_0}-\E_A^{r_1}\right\|_\diamond.
\end{align*}
\begin{figure}[t]
  \centering
  \includegraphics[width=\columnwidth]{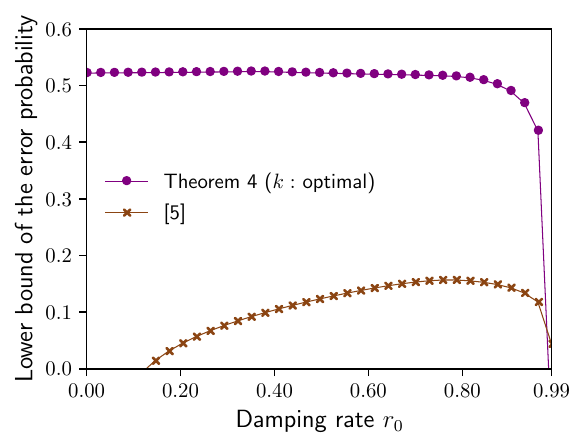}
  \caption{
    The lower bounds on the channel position finding problem for amplitude damping channels
    ($\ell=3, n=15, p_\xi=1/\ell\ (\xi\in\Xi), r_1=r_0+0.01$);
    $\bullet$, the bound from Theorem~\ref{boundForQCGDPWithWeightedTraceDistance}
    where $k$ is optimal;
    $\times$, the bound from~\cite{zhuang2020ultimate}.
  }
  \label{fig:multiple_qadc_r0}
\end{figure}
\begin{figure}[t]
  \centering
  \includegraphics[width=\columnwidth]{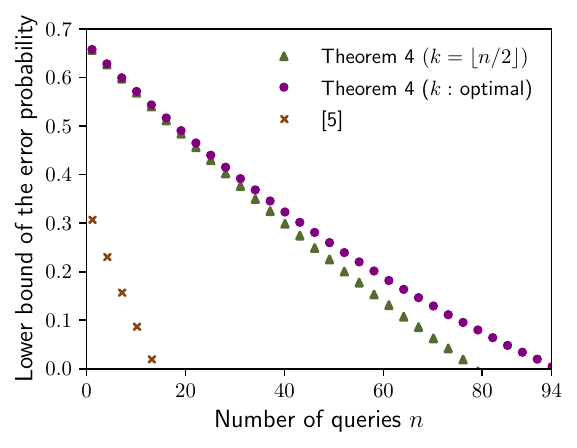}
  \caption{
    The lower bounds on the channel position finding problem for amplitude damping channels
    ($\ell=3, p_\xi=1/\ell\ (\xi\in\Xi), r_0=0.10, r_1=0.11$);
    $\blacktriangle$, the bound from Theorem~\ref{boundForQCGDPWithWeightedTraceDistance}
    where $k=\lfloor n/2\rfloor$;
    $\bullet$, the bound from Theorem~\ref{boundForQCGDPWithWeightedTraceDistance}
    where $k$ is optimal;
    $\times$, the bound from~\cite{zhuang2020ultimate}.
  }
  \label{fig:multiple_qadc_n}
\end{figure}
The numerical results for the channel position finding problem for amplitude damping channels
are shown in Figs.~\ref{fig:multiple_qadc_r0}
and~\ref{fig:multiple_qadc_n}.
The distribution of the channels is uniform, i.e., $p_\xi=1/\ell$ for $\xi\in\Xi$.
The number of oracles is $\ell=3$.
In Fig.~\ref{fig:multiple_qadc_r0}, the damping rates of the channels 
are $r_0$ and $r_1=r_0+0.01$.
In Fig.~\ref{fig:multiple_qadc_r0}, the number of queries is $n=15$.
The parameters $\alpha_0$ and $\alpha_1,$ are optimized to maximize the lower bound.
In Fig.~\ref{fig:multiple_qadc_n},
the damping rates of the channels are $r_0=0.10$ and $r_1=0.11$.
Compared with a lower bound obtained by Zhuang and Pirandola~\cite{zhuang2020ultimate},
the lower bound given by Theorem~\ref{boundForQCGDPWithWeightedTraceDistance}
is better for all $r_0$ in Figs.~\ref{fig:multiple_qadc_r0}
and for all $n$ in Fig.~\ref{fig:multiple_qadc_n}.

\section{Proofs of the Theorems}\label{sec:proof}
\subsection{Idea of the proofs}
For describing the ideas of the proofs, we consider a simple case $\qcdp{1/2,1/2}{\mathcal{O}_{A\to B}^0,\mathcal{O}_{A\to B}^1}$ in this section.
  Let $((\Phi_{BR\to AR}^j),(M_{BR}^\eta))$
  be a discrimination algorithm that achieves the minimum error probability
  $\perr(n)$ for this problem with $n$ queries.
From Proposition~\ref{holevoHelstrom}, the error probability of the algorithm is
\begin{equation*}
\frac{1}{2}\left(1-\frac12\left\|\rho_{BR}^0-\rho_{BR}^1\right\|_1\right),
\end{equation*}
where
\begin{align}
\rho_{BR}^0 &:= \mathcal{O}_{A\to B}^0 \circ \Phi_{BR\to AR}^n\circ
\mathcal{O}_{A\to B}^0 \circ \Phi_{BR\to AR}^{n-1}\circ\dotsm\nonumber\\
&\qquad\circ \mathcal{O}_{A\to B}^0 \circ \Phi_{BR\to AR}^{1}\left(\ketbra{0}{0}_{BR}\right)\nonumber\\
\rho_{BR}^1 &:= \mathcal{O}_{A\to B}^1 \circ \Phi_{BR\to AR}^n\circ
\mathcal{O}_{A\to B}^1 \circ \Phi_{BR\to AR}^{n-1}\circ\dotsm\nonumber\\
&\qquad\circ \mathcal{O}_{A\to B}^1 \circ \Phi_{BR\to AR}^{1}\left(\ketbra{0}{0}_{BR}\right).
\label{eq:rho01}
\end{align}
We use the well-known hybrid argument for upper bounding $\frac12\|\rho_{BR}^0-\rho_{BR}^1\|_1$~\cite{bennett1997strengths}.
Let
\begin{align*}
&\rho_{BR}^{(j)} := 
\mathcal{O}_{A\to B}^1 \circ \Phi_{BR\to AR}^n\circ
\dotsm \circ \mathcal{O}_{A\to B}^1 \circ \Phi_{BR\to AR}^{j+1}\\
&\circ\mathcal{O}_{A\to B}^0 \circ \Phi_{BR\to AR}^j\circ
\dotsm \circ \mathcal{O}_{A\to B}^0 \circ \Phi_{BR\to AR}^{1}\left(\ketbra{0}{0}_{BR}\right),
\end{align*}
for $j=0,1,\dotsc,n$.
In the definition of $\rho_{BR}^{(j)}$, we invoke $\mathcal{O}_{A\to B}^0$ for the first $j$ oracle calls, and $\mathcal{O}_{A\to B}^1$ for the last $n-j$ oracle calls.
Then, $\rho_{BR}^{(0)} = \rho_{BR}^1$ and $\rho_{BR}^{(n)} = \rho_{BR}^0$.
Hence,
\begin{align*}
\frac12\|\rho_{BR}^0-\rho_{BR}^1\|_1
&= \frac12\|\rho_{BR}^{(n)}-\rho_{BR}^{(0)}\|_1\\
&\le \sum_{j=0}^{n-1}\frac12\left\|\rho_{BR}^{(j+1)}-\rho_{BR}^{(j)}\right\|_1.
\end{align*}
Here, we used the triangle inequality of the trace distance.
Hence, it is sufficient to upper bound $\frac12\|\rho_{BR}^{(j+1)}-\rho_{BR}^{(j)}\|_1$ for each $j=0,1,\dotsc,n-1$.
For each $j=0,1,\dotsc,n-1$, we obtain
\begin{align*}
&\frac12\left\|\rho_{BR}^{(j+1)}-\rho_{BR}^{(j)}\right\|_1\\
&=
\frac12\Bigl\|
\mathcal{O}_{A\to B}^1 \circ \Phi_{BR\to AR}^n\circ
\dotsm \circ \mathcal{O}_{A\to B}^1 \circ \Phi_{BR\to AR}^{j+2}\\
&\quad\circ\mathcal{O}_{A\to B}^0 \circ \Phi_{BR\to AR}^{j+1}\circ
\dotsm \circ \mathcal{O}_{A\to B}^0 \circ \Phi_{BR\to AR}^{1}\left(\ketbra{0}{0}_{BR}\right)\\
&\quad- \mathcal{O}_{A\to B}^1 \circ \Phi_{BR\to AR}^n\circ
\dotsm \circ \mathcal{O}_{A\to B}^1 \circ \Phi_{BR\to AR}^{j+1}\\
&\quad\circ\mathcal{O}_{A\to B}^0 \circ \Phi_{BR\to AR}^j\circ
\dotsm \circ \mathcal{O}_{A\to B}^0 \circ \Phi_{BR\to AR}^{1}\left(\ketbra{0}{0}_{BR}\right)\Bigr\|_1\\
&\le
\frac12\Bigl\|
\mathcal{O}_{A\to B}^0 \circ \Phi_{BR\to AR}^{j+1}\circ
\dotsm \circ \mathcal{O}_{A\to B}^0 \circ \Phi_{BR\to AR}^{1}\left(\ketbra{0}{0}_{BR}\right)\\
&\quad-\mathcal{O}_{A\to B}^1\circ\Phi_{BR\to AR}^{j+1}\circ\mathcal{O}_{A\to B}^0 \circ \Phi_{BR\to AR}^j\\
&\quad\circ \dotsm \circ \mathcal{O}_{A\to B}^0 \circ \Phi_{BR\to AR}^{1}\left(\ketbra{0}{0}_{BR}\right)\Bigr\|_1\\
&\le
\max_{\rho_{AR}}\frac12\left\|
\mathcal{O}_{A\to B}^0(\rho_{AR})
-\mathcal{O}_{A\to B}^1(\rho_{AR})\right\|_1.
\end{align*}
In the first inequality, we used the monotonicity of the trace distance, i.e., $\frac12\|\Phi(\rho)-\Phi(\sigma)\|_1\le\frac12\|\rho-\sigma\|_1$
for any quantum channel $\Phi$ and quantum states $\rho$ and $\sigma$.
In the second inequality, we replace the quantum state 
\begin{equation*}
\Phi_{BR\to AR}^{j+1}\circ\mathcal{O}_{A\to B}^0 \circ \dotsm\circ\mathcal{O}_{A\to B}^0 \circ \Phi_{BR\to AR}^{1}\left(\ketbra{0}{0}_{BR}\right),
\end{equation*}
with a general quantum state $\rho_{AR}$.
Hence,
\begin{align}
\frac12\|\rho_{BR}^0-\rho_{BR}^1\|_1
\le
\frac{n}2\max_{\rho_{AR}}\left\|
\mathcal{O}_{A\to B}^0(\rho_{AR})
-\mathcal{O}_{A\to B}^1(\rho_{AR})\right\|_1.
\label{eq:trivialub}
\end{align}
Note that the maximum of the trace distance in~\eqref{eq:trivialub} is equal to the diamond distance $\|\mathcal{O}_{A\to B}^0-\mathcal{O}_{A\to B}^1\|_\diamond$
when $\dim(R)\ge\dim(A)$.

In this paper, we improve the upper bound~\eqref{eq:trivialub} of the trace distance.
We will employ two ideas to improve the upper bound.
In the first idea, we use another distance measure, namely the Bures angle (Theorem~\ref{boundForQCDPWithBuresAngle}).
In the second idea, we introduce weights for the series of quantum states $(\rho_{BR}^{(j)})_j$ (Theorem~\ref{boundForQCDPWithWeightedTraceDistance}).
Regarding to the first idea, the bound~\eqref{eq:trivialub} can be straightforwardly generalized to
 arbitrary distance measure $D$ satisfying the triangle inequality and the monotonicity as
\begin{align}
D(\rho_{BR}^0,\rho_{BR}^1)
\le
n\max_{\rho_{AR}}
D\left(\mathcal{O}_{A\to B}^0(\rho_{AR}), \mathcal{O}_{A\to B}^1(\rho_{AR})\right).
\label{eq:gtrivialub}
\end{align}
We will use this inequality where $D$ is the Bures angle.

\subsection{Proof of Theorem~\ref{boundForQCDPWithBuresAngle}}
In the proof of Theorem~\ref{boundForQCDPWithBuresAngle}, we use the Bures angle in place of the trace distance.
For bounding the trace distance by using the Bures angle or equivalently fidelity, we use the well-known Fuchs--van~de~Graaf inequality
$\|\rho_A^0-\rho_A^1\|_1\leq 2\sqrt{1-\F{\rho_A^0}{\rho_A^1}^2}$~\cite{fuchs1999cryptographic}. 
The following lemma is a minor generalization of the Fuchs--van~de~Graaf inequality.
\begin{lemma}\label{fuchsVanDeGraaf}
  Let $a_0,a_1\in[0,\infty)$ be non-negative real numbers.
  Let $\rho_A^0,\rho_A^1\in\mathsf{D}(A)$ be density operators.
  The following then holds:
  \[
    \left\|a_0\rho_A^0-a_1\rho_A^1\right\|_1 
    \leq\sqrt{(a_0+a_1)^2-4a_0a_1\F{\rho_A^0}{\rho_A^1}^2}. 
  \]
The equality holds if $\rho_A^0$ and $\rho_A^1$ are pure states.
\end{lemma}
The proof is given in Appendix~\ref{apx:fuchsVanDeGraaf}.

From Proposition~\ref{holevoHelstrom},
the error probability of the discrimination algorithm is
\begin{equation}
\frac{1}{2}\left(1-\left\|p_0\rho_{BR}^0-p_1\rho_{BR}^1\right\|_1\right).
\label{eq:AHH}
\end{equation}
where $\rho_{BR}^0$ and $\rho_{BR}^1$ are defined in~\eqref{eq:rho01}.
From Lemma~\ref{fuchsVanDeGraaf},
\begin{equation}
\left\|p_0\rho_{BR}^0-p_1\rho_{BR}^1\right\|_1 
\leq\sqrt{1-4p_0p_1\F{\rho_{BR}^0}{\rho_{BR}^1}^2}. 
\label{eq:gFV}
\end{equation}
Hence, our goal is to lower bound $F(\rho_{BR}^0,\rho_{BR}^1)$.
Equivalentlly, our goal is to upper bound $\bures{\rho_{BR}^0}{\rho_{BR}^1}$.
From~\eqref{eq:gtrivialub} for $D=\mathcal{A}$, we obtain
\begin{equation*}
\bures{\rho_{BR}^0}{\rho_{BR}^1}
\le
n\max_{\rho_{AR}}
\bures{\mathcal{O}_{A\to B}^0(\rho_{AR})}{\mathcal{O}_{A\to B}^1(\rho_{AR})}.
\end{equation*}
Since
\begin{align*}
&\max_{\rho_{AR}} \bures{\mathcal{O}_{A\to B}^0(\rho_{AR})}{\mathcal{O}_{A\to B}^1(\rho_{AR})}\\
&=
\arccos \min_{\rho_{AR}} F\left(\mathcal{O}_{A\to B}^0(\rho_{AR}), \mathcal{O}_{A\to B}^1(\rho_{AR})\right)
=\tau_\mathcal{A},
\end{align*}
we obtain
\begin{equation}
F(\rho_{BR}^0,\rho_{BR}^1)= \cos \bures{\rho_{BR}^0}{\rho_{BR}^1}
\ge \cos n\tau_\mathcal{A},
\label{eq:FA}
\end{equation}
if $n\tau_\mathcal{A}\le \pi/2$.
Hence, we obtain~\eqref{eq:A2} from~\eqref{eq:AHH},~\eqref{eq:gFV} and~\eqref{eq:FA}.

Kawachi et al.\ characterized the exact minimum error probability
when the two quantum channels are both unitary channels, and 
the distribution of the oracle is uniform (i.e., $p_0=p_1=0.5$)
in~\cite{kawachi2019quantum}. 
Their lower bound can be reproduced from Theorem~\ref{boundForQCDPWithBuresAngle}. 
The minimum error probability is represented by using the \textit{covering angle}.
\begin{definition}[Covering angle]
  Let $\arg_{\geq 0}(z)\coloneqq\min\{\varphi\geq 0\mid z=|z|\exp(i\varphi)\}$
  for a complex number $z$.
  Let $\theta_1,\theta_2,\dotsc,\theta_n$ be real numbers and
  $\Theta\coloneqq\{\exp(i\theta_1),\exp(i\theta_2),\dotsc,\exp(i\theta_n)\}$.
  The covering angle $\cov$ of $\Theta$ is defined as
  \begin{align*}
    \cov&\coloneqq\min_{j\in\{1,2,\dotsc,n\}}\max_{\ell\in\{1,2,\dotsc,n\}}
    \arg_{\geq 0}\Bigl(\exp\bm(i(\theta_\ell-\theta_j)\bm)\Bigr).
  \end{align*}
\end{definition}

From Theorem~\ref{boundForQCDPWithBuresAngle} and an algorithm in~\cite{kawachi2019quantum},
we obtain the following corollary.
\begin{corollary}\label{boundForUODP}
  Let $O_A^0$ and $O_A^1$ be unitary operators.
  Let $\mathcal{O}_A^0,\mathcal{O}_A^1$ be quantum channels defined as
  $\mathcal{O}_A^\xi(\rho_A)\coloneqq O_A^\xi\rho_AO_A^{\xi\dagger}$ for $\xi\in\{0,1\}$. 
  Let $\perr(n)$ be the minimum error probability 
  for $\qcdp{p_0,p_1}{\mathcal{O}_A^0,\mathcal{O}_A^1}$ with $n$ queries.
  Let $\cov$ be the covering angle of the set of eigenvalues
  of an operator $O_A^{0\dagger}O_A^1$.
  The following then holds:
  \begin{align*}
    \perr(n)=\dfrac{1}{2}\left(1-\sqrt{1-4p_0p_1\cos^2\dfrac{n\cov}{2}}\right),
  \end{align*}
  under the condition $n\cov\le\pi$, and $\perr(n)=0$ if $n\cov>\pi$.
\end{corollary}
The lower bound is obtained from~\eqref{eq:A2} and the equality $\tau_\A=\cov/2$ 
(see~\cite{kawachi2019quantum} for the details).
The upper bound is obtained from Lemma~\ref{fuchsVanDeGraaf} 
and the explicit non-adaptive algorithm in~\cite{kawachi2019quantum} as explained below.
Let $\ket{\psi_0}_A$ and $\ket{\psi_1}_A$ be eigenvectors of the unitary matrix $O_A^{0\dagger}O_A^1$ determining the covering angle
$\theta_\mathrm{cover}$ i.e.,
\begin{align*}
O_A^{0\dagger}O_A^1\ket{\psi_0}_A&=\exp(i\theta_0)\ket{\psi_0}_A\\
O_A^{0\dagger}O_A^1\ket{\psi_1}_A&=\exp(i\theta_1)\ket{\psi_1}_A\\
\theta_\mathrm{cover}&=\arg_{\ge0}(\exp(i(\theta_0-\theta_1))).
\end{align*}
In the optimal non-adaptive algorithm in~\cite{kawachi2019quantum},
first prepare the initial state $\bigl(\ket{\psi_0}_A^{\otimes n}+\ket{\psi_1}_A^{\otimes n}\bigr)/\sqrt{2}$,
 and then apply $n$ oracles $\mathcal{O}^{\otimes n}$ in parallel.
Then, the quantum states just before measurement are
\begin{align*}
&O_A^{0\otimes n}\frac{\ket{\psi_0}_A^{\otimes n}+\ket{\psi_1}_A^{\otimes n}}{\sqrt{2}}\qquad\text{and}\\
&O_A^{1\otimes n}\frac{\ket{\psi_0}_A^{\otimes n}+\ket{\psi_1}_A^{\otimes n}}{\sqrt{2}}
\end{align*}
The fidelity (the absolute value of the inner product) of the state vectors is
$|\exp(in\theta_0) + \exp(in\theta_1)|/2 = \cos(n\theta_\mathrm{cover}/2)$ if $n\theta_\mathrm{cover}\le\pi$.
Hence, this non-adaptive algorithm achieves the minimum error probability $\perr(n)$ in Corollary~\ref{boundForUODP}.
This non-adaptive algorithm requires a working system $R$ with $\dim(R)=\dim(A)^{n-1}$.
There also exists an optimal adaptive algorithm that does not require ancillas, i.e., $\dim(R)=0$.
In the algorithm, first prepare the initial state $\bigl(\ket{\psi_0}_A+\ket{\psi_1}_A\bigr)/\sqrt{2}$,
and then apply the oracle $n$ times in serial.
This algorithm achieves the minimum error probability as well.

Although $\tau_\A$ is represented using the minimum of 
the concave function, when $\dim(R)\geq\dim(A)$,
$\tau_\A$ can be represented by the minimum of the convex function~\eqref{eq:ft}.
This is proved in Appendix~\ref{apx:ft}.

\subsection{Proof of Theorem~\ref{boundForQCDPWithWeightedTraceDistance}}
Theorem~\ref{boundForQCDPWithWeightedTraceDistance} is obtained similarly to~\eqref{eq:trivialub}, but using weights for quantum states.
We first show the idea of the technique briefly.
Let us consider a distance measure induced from some norm $\|\cdot\|$.
It holds a triangle inequality $\|u-v\|\le\|u-x\|+\|x-v\|$.
Here, we may obtain a better upper bound $\|u-v\|\le\|u-r x\|+\|r x-v\|$ where $r\in(0,1]$.
We can intuitively understand it for the case of the Euclidean norm as in Fig.~\ref{fig:tri}.

\begin{figure}
  \includegraphics{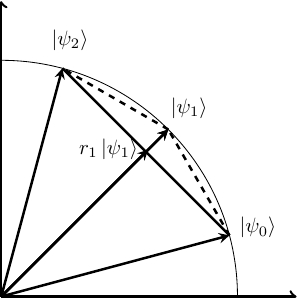}
  \caption{Triangle inequality for the Euclidean distance.}
  \label{fig:tri}
\end{figure}  
\begin{figure}
  \includegraphics{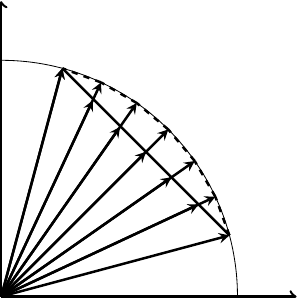}
  \caption{The shortcut method.}
  \label{fig:tris}
\end{figure}
Because there are $n-1$ internal quantum states $\rho_{BR}^{(1)},\dotsc,\rho_{BR}^{(n-1)}$ in our proof,
we have to determine $n-1$ ``weights'' $r_{1},\dotsc,r_{n-1}$, as in Fig.~\ref{fig:tris}.
Because the optimization of $n-1$ parameters is computationally 
hard if $n$ is large, we only use three parameters $\alpha_0,\,\alpha_1\in(0,1]$,
and $k\in\{1,2,\dotsc,n\}$ such that $\alpha_0^{k}=\alpha_1^{n-k}$,
and assume that $r_{n-i}=\alpha_0^i$ if $i\le k$ and $r_{i}=\alpha_1^i$ if $i\le n-k$.
We then need to optimize only two parameters $\alpha_0$ and $k$.
Once the parameters are fixed, we can calculate the lower bound in a constant time with respect to the number $n$ of queries.

In the same way as for the proof of Theorem~\ref{boundForQCDPWithBuresAngle},
we obtain the following inequality:
\begin{align*}
  &\frac12\left\|p_0\rho_{BR}^{(n)}-p_1\rho_{BR}^{(0)}\right\|_1\\
  &\leq\sum_{j=0}^{k-1}\frac12\left\|
    p_0\alpha_0^j\rho_{BR}^{(n-j)}-p_0\alpha_0^{j+1}\rho_{BR}^{(n-j-1)}
  \right\|_1\\
  &\quad+\sum_{j=0}^{n-k-1}\frac12\left\|
    p_1\alpha_1^{j+1}\rho_{BR}^{(j+1)}-p_1\alpha_1^j\rho_{BR}^{(j)}
  \right\|_1.
\end{align*}
We used the condition $p_0\alpha_0^k=p_1\alpha_1^{n-k}$ in the above inequality.
For the first term of the right-hand side, we obtain
\begin{align*}
  &\sum_{j=0}^{k-1}\frac12\left\|
    p_0\alpha_0^j\rho_{BR}^{(n-j)}-p_0\alpha_0^{j+1}\rho_{BR}^{(n-j-1)}
  \right\|_1\\
  &\leq p_0\sum_{j=0}^{k-1}\frac{\alpha_0^j}2\Bigl\|
    \mathcal{O}_{A\to B}^0\circ\Phi_{BR\to AR}^{j+1}\circ\dotsm\circ\Phi_{BR\to AR}^1\left(\ketbra{0}{0}_{BR}\right)\\*
  &\hspace{2em}
    -\alpha_0\mathcal{O}_{A\to B}^1\circ\Phi_{BR\to AR}^{j+1}\circ\dotsm\circ\Phi_{BR\to AR}^1\left(\ketbra{0}{0}_{BR}\right)
  \Bigr\|_1\\
  &\leq p_0\left(\sum_{j=0}^{k-1}\alpha_0^j\right)\max_{\rho_{AR}\in\mathsf{D}(AR)}\frac{1}{2}\Bigl\|
    \left(\mathcal{O}_{A\to B}^0-\alpha_0\mathcal{O}_{A\to B}^1\right)(\rho_{AR})
  \Bigr\|_1\\
  &=p_0\left(\sum_{j=0}^{k-1}\alpha_0^j\right)\tau_\diamond^0.
\end{align*}
The first inequality follows from the monotonicity of the trace norm.
The equality follows from the convexity of the trace distance.

Similarly, we obtain
\begin{align*}
  &\sum_{j=0}^{n-k-1}\frac12\left\|
    p_1\alpha_1^{j+1}\rho_{BR}^{(j+1)}-p_1\alpha_1^j\rho_{BR}^{(j)}
  \right\|_1\\
  &\leq p_1\left(\sum_{j=0}^{n-k-1}\alpha_1^j\right)\tau_\diamond^1.
\end{align*}
Therefore, we obtain
\begin{align*}
  &\perr(n)\\
  &=\frac{1}{2}
    -\frac{1}{2}\left\|p_0\rho_{BR}^{0,n}-p_1\rho_{BR}^{n,n}\right\|_1\\
  &\ge\frac{1}{2}
    -p_0\left(\sum_{j=0}^{k-1}\alpha_0^j\right)\tau_\diamond^0
    -p_1\left(\sum_{j=0}^{n-k-1}\alpha_1^j\right)\tau_\diamond^1.
\end{align*}
This proves Theorem~\ref{boundForQCDPWithWeightedTraceDistance}.

\subsection{Proof of Theorem~\ref{boundForQCGDPWithBuresAngle}}
Theorems~\ref{boundForQCGDPWithBuresAngle} and~\ref{boundForQCGDPWithWeightedTraceDistance} are also proved using the triangle inequalities for the Bures angle and
the trace distance with weights, respectively.
Let $((\Phi_{BR\to AR}^j),(M_{BR}^\eta))$
be a discrimination algorithm with $n$ queries that achieves the minimum error probability 
$\perr(n)$ of $\qcgdp{p_\xi}{\mathcal{O}_{A\to B}^\xi}{C_\eta}$.
Let $W$ be a two-dimensional quantum system. 
Quantum channels $(\M_{BR\to W}^\xi)_{\xi\in\Xi}$ are defined as
\begin{align*}
  \M_{BR\to W}^\xi(\rho_{BR})
  &\coloneqq\Tr\left(\sum_{\eta\colon\xi\notin C_\eta}M_{BR}^\eta\rho_{BR}\right)\Ketbra{0}{0}_W\\
  &\quad +\Tr\left(\sum_{\eta\colon\xi\in C_\eta}M_{BR}^\eta\rho_{BR}\right)\Ketbra{1}{1}_W,
\end{align*}
for each $\xi\in\Xi$. 
The quantum channel $\M_{BR\to W}^\xi$ represents the following procedure:
One performs POVM $(M_{BR}^\eta)_{\eta\in H}$ to the input state $\rho_{BR}$,
and then outputs $\ketbra{0}{0}_W$ if the measurement outcome 
$\eta\in H$ satisfies $\xi\notin C_\eta$, 
and outputs $\ketbra{1}{1}_W$ otherwise.

Let $\Psi_{AR\to BR}$ be an arbitrary quantum channel.
We define density operators as follows:
\begin{align*}
\rho_{BR}^{(j,\xi)}&\coloneqq \mathcal{O}_{A\to B}^\xi\circ\Phi_{BR\to AR}^n\circ\dotsm \circ \mathcal{O}_{A\to B}^\xi \circ\Phi_{BR\to AR}^{j+1}\\
&\circ\Psi_{AR\to BR}\circ\Phi_{BR\to AR}^j\circ\dotsm\\
&\circ\Psi_{AR\to BR}\circ\Phi_{BR\to AR}^1\left(\ketbra{0}{0}_{BR}\right)\\
\sigma_W^{(t)} &\coloneqq 
\sum_{\xi\in\Xi}p_\xi\M_{BR\to W}^\xi(\rho_{BR}^{(n-t,\xi)})
\hspace{2.7em}(0\leq i\leq n).
\end{align*}
The density operator $\sigma_W^{(t)}$ is obtained by the following procedure:
First, one picks $\xi$ according to the probability distribution $(p_\xi)_\xi$.
Then, one performs the discrimination algorithm in which the first $n-t$ queries are given to $\Psi_{AR\to BR}$ in place of $\mathcal{O}_{A\to B}^\xi$.
Finally, one performs the measurement $(M_{BR}^\eta)_\eta$ and outputs $\ketbra{0}{0}_W$ or $\ketbra{1}{1}_W$ if $\xi\notin C_\eta$ or $\xi\in C_\eta$, respectively.
Hence, the weights of $\ketbra{0}{0}_W$ and $\ketbra{1}{1}_W$ in $\sigma_W^{(n)}$ are the probability of failure and success of the modified discrimination algorithm, respectively.

By applying the triangle inequality for the Bures angle to 
a sequence of density operators
$\ketbra{1}{1}_W,\,\sigma_W^{(n)},\,\sigma_W^{(n-1)},\,\dotsc,\,\sigma_W^{(m)},\,\ketbra{0}{0}_W$,
we obtain the following inequality:
\begin{align*}
  \frac{\pi}{2}
  &=\bures{\Ketbra{1}{1}_W}{\Ketbra{0}{0}_W}\\
  &\leq\bures{\Ketbra{1}{1}_W}{\sigma_W^{(n)}}
  +\sum_{t=m+1}^n\bures{\sigma_W^{(t)}}{\sigma_W^{(t-1)}}\\*
  &\quad+\bures{\sigma_W^{(m)}}{\Ketbra{0}{0}_W}.
\end{align*}
Each term in the upper bound is evaluated 
as follows. For the first term of the upper bound, we obtain
\begin{align*}
  &\bures{\Ketbra{1}{1}_W}{\sigma_W^{(n)}}\\
  &=\arccos\sqrt{\Braket{1|\sigma_W^{(n)}|1}_W}\\
  &=\arccos\sqrt{
    \sum_{\xi\in\Xi}p_\xi\sum_{\eta\colon\xi\in C_\eta}
    \Tr\left(M^\eta_{BR}\rho_{BR}^{(0,\xi)}\right)
  }\\
  &=\arccos\sqrt{1-\perr(n)}.
\end{align*}
For the second term of the upper bound, we obtain
\begin{align*}
  &\bures{\sigma_W^{(t)}}{\sigma_W^{(t-1)}}\\
  &=\arccos F\Biggl(\sum_{\xi\in\Xi}p_\xi\M_{BR\to W}^\xi\left(\rho_{BR}^{(n-t,\xi)}\right),\\*
  &\hspace{6em}
  \sum_{\xi\in\Xi}p_\xi\M_{BR\to W}^\xi\left(\rho_{BR}^{(n-t+1,\xi)}\right)\Biggr)\\
  &\leq\arccos\Biggl(\sum_{\xi\in\Xi}p_\xi
  F\Bigl(\M_{BR\to W}^\xi\left(\rho_{BR}^{(n-t,\xi)}\right),\\*
  &\hspace{6em}
  \M_{BR\to W}^\xi\left(\rho_{BR}^{(n-t+1,\xi)}\right)\Bigr)\Biggr)\\
  &\leq\arccos\Biggl(\sum_{\xi\in\Xi}p_\xi
  F\Bigl(\mathcal{O}_{A\to B}^\xi\circ\Phi_{BR\to AR}^{n-t+1}\circ\Psi_{AR\to BR}\\
  &\quad\circ\Phi_{BR\to AR}^{n-t}\circ\dotsm \circ\Psi_{AR\to BR}\circ\Phi_{BR\to AR}^1\left(\ketbra{0}{0}_{BR}\right),\\*
  &\quad\Psi_{AR\to BR}\circ\Phi_{BR\to AR}^{n-t+1}\circ\dotsm\circ\Psi_{AR\to BR}\\*
  &\quad\circ\Phi_{BR\to AR}^1\left(\ketbra{0}{0}_{BR}\right)\Bigr)\Biggr)\\
  &\leq\arccos\min_{\rho_{AR}\in\mathsf{D}(AR)}\sum_{\xi\in\Xi}p_\xi
  F\Bigl(\mathcal{O}_{A\to B}^\xi(\rho_{AR}),\\*
  &\hspace{6em}\Psi_{AR\to BR}(\rho_{AR})\Bigr)\\
  &=\theta_\A(\Psi_{AR\to BR}).
\end{align*}
The first inequality and the last equality follow from 
the joint concavity of the fidelity.
The second inequality follows from the monotonicity of the fidelity.
In the third inequality, we replace the quantum state
\begin{equation*}
  \Phi_{BR\to AR}^{n-t+1}\circ\dotsm\circ\Psi_{AR\to BR}\circ\Phi_{BR\to AR}^1\left(\ketbra{0}{0}_{BR}\right)
\end{equation*}
with a general quantum state $\rho_{AR}$.

For the last term, we obtain
\begin{align*}
  &\bures{\sigma_W^{(m)}}{\Ketbra{0}{0}_W}\\
  &=\arccos\sqrt{\Braket{0|\sigma_W^{(m)}|0}_W}\\
  &=\arccos\sqrt{
    \sum_{\xi\in\Xi}p_\xi\sum_{\eta\colon\xi\notin C_\eta}
    \Tr\left(M^\eta_{BR}\rho_{BR}^{(n-m,\xi)}\right)
  }\\
  &\leq\arccos\sqrt{\perr(m)}.
\end{align*}
From the above inequalities, we obtain
\begin{align*}
  \frac{\pi}{2}
  &\leq\arccos\sqrt{1-\perr(n)}\\*
  &\quad+(n-m)\theta_\A(\Psi_{AR\to BR})+\arccos\sqrt{\perr(m)}.
\end{align*}
Therefore, we obtain
\begin{align*}
  &1-\perr(n)\\
  &\leq\cos^2\left(
    \frac{\pi}{2}-(n-m)\theta_\A(\Psi_{AR\to BR})-\arccos\sqrt{\perr(m)}
  \right)\\
  &=\sin^2\bm((n-m)\theta_\A(\Psi_{AR\to BR})+\theta_m\bm),
\end{align*}
for any $\Psi_{AR\to BR}$ and $m$ satisfying 
$(n-m)\theta_\A(\Psi_{AR\to BR})+\theta_m\in[0,\pi/2]$.
Hence, we obtain the main part of Theorem~\ref{boundForQCGDPWithBuresAngle}.

Although $\theta_\A$ is represented using the minimum of 
the concave function, when $\dim(R)\geq\dim(A)$ and $\Psi_{AR\to BR}=\Psi_{A\to B}\otimes \mathrm{id}_R$ for some quantum channel $\Psi_{A\to B}$,
$\theta_\A$ can be represented by the minimum of a convex function~\eqref{eq:ft}.
\begin{lemma}\label{minimumFidelityFormula}
  Let $\mathcal{O}_{A\to B}^\xi$ and $\Psi_{A\to B}$ be quantum channels with 
  Stinespring representations $O_{A\to BE}^\xi$ and $V_{A\to BE}$, respectively.
  If $\dim(R)\geq\dim(A)$, the following holds:
  \begin{align*}
    &\min_{\ket{\phi}_{AR}\in\mathsf{S}(AR)}\sum_{\xi\in\Xi}p_\xi
      \F{\mathcal{O}_{A\to B}^\xi(\ketbra{\phi}{\phi}_{AR})}{\Psi_{A\to B}(\ketbra{\phi}{\phi}_{AR})}\\
    &=\min_{\sigma_A\in\mathsf{D}(A)}\sum_{\xi\in\Xi}p_\xi\left\|
      \Tr_B\left(O_{A\to BE}^\xi\sigma_{A}(V_{A\to BE})^\dagger\right)
    \right\|_1.
  \end{align*}
\end{lemma}
The proof is the same as the proof of Lemma~\ref{lem:ft} in Appendix~\ref{apx:ft}.

\subsection{Proof of Theorem~\ref{boundForQCGDPWithWeightedTraceDistance}}
Theorem~\ref{boundForQCGDPWithWeightedTraceDistance} is proved similarly to Theorem~\ref{boundForQCGDPWithBuresAngle}.
The quantum states $\phi_{BR}^{(j,\xi)}$ and $\sigma_W^{(t)}$ are defined in the same way.
We obtain the following inequality:
\begin{align*}
  1
  &=\frac{1}{2}\left\|\Ketbra{1}{1}_W-\Ketbra{0}{0}_W\right\|_1\\
  &\leq\frac{1}{2}\left\|\Ketbra{1}{1}_W-\sigma_W^{(n)}\right\|_1\\
  &\quad+\sum_{j=0}^{n-k-1}\frac{1}{2}\left\|
    \alpha_0^j\sigma_W^{(n-j)}-\alpha_0^{j+1}\sigma_W^{(n-j-1)}
  \right\|_1\\
  &\quad+\sum_{j=0}^{k-m-1}\frac{1}{2}\left\|
    \alpha_1^{j+1}\sigma_W^{(m+j+1)}-\alpha_1^j\sigma_W^{(m+j)}
  \right\|_1\\
  &\quad+\frac{1}{2}\left\|\sigma_W^{(m)}-\Ketbra{0}{0}_W\right\|_1.
\end{align*}
We used the condition $\alpha_0^{n-k}=\alpha_1^{k-m}$ in the above inequality.
Each term on the upper bound is evaluated 
as follows. For the first term, we obtain
\begin{align*}
  &\frac{1}{2}\left\|\Ketbra{1}{1}_W-\sigma_W^{(n)}\right\|_1\\
  &=\sum_{\xi\in\Xi}p_\xi\sum_{\eta\colon\xi\notin C_\eta}
    \Tr\left(M^\eta_{BR}\rho_{BR}^{(0,\xi)}\right)\\
  &=\perr(n).
\end{align*}

For the second term, we obtain
\begin{align*}
  &\sum_{j=0}^{n-k-1}\frac{1}{2}\left\|
    \alpha_0^j\sigma_W^{(n-j)}-\alpha_0^{j+1}\sigma_W^{(n-j-1)}
  \right\|_1\\
  &=\sum_{j=0}^{n-k-1}\frac{\alpha_0^j}{2}\,\Biggl\|
    \sum_{\xi\in\Xi}p_\xi\M_{BR\to W}^\xi\left(\rho_{BR}^{(j,\xi)}\right)\\*
  &\hspace{7em}
    -\alpha_0\sum_{\xi\in\Xi}p_\xi\M_{BR\to W}^\xi\left(\rho_{BR}^{(j+1,\xi)}\right)
  \Biggr\|_1\\
  &\leq\sum_{j=0}^{n-k-1}\frac{\alpha_0^j}{2}\sum_{\xi\in\Xi}p_\xi\,\Biggl\|
    \M_{BR\to W}^\xi\left(\rho_{BR}^{(j,\xi)}\right)\\*
  &\hspace{7em}
    -\alpha_0\M_{BR\to W}^\xi\left(\rho_{BR}^{(j+1,\xi)}\right)
  \Biggr\|_1\\
  &\leq\sum_{j=0}^{n-k-1}\frac{\alpha_0^j}{2}\sum_{\xi\in\Xi}p_\xi\,\Biggl\|
    \mathcal{O}_{A\to B}^\xi\circ\Phi_{BR\to AR}^{j+1}\\
  &\qquad\circ\Psi_{AR\to BR}\circ\Phi_{BR\to AR}^j\circ\dotsm\circ\Psi_{AR\to BR}\\*
  &\qquad\circ\Phi_{BR\to AR}^1\left(\ketbra{0}{0}_{BR}\right)
    -\alpha_0\Psi_{AR\to BR}\circ\Phi_{BR\to AR}^{j+1}\\
  &\qquad\circ\Psi_{AR\to BR}\circ\Phi_{BR\to AR}^j\circ\dotsm\circ\Psi_{AR\to BR}\\*
  &\qquad\circ\Phi_{BR\to AR}^1\left(\ketbra{0}{0}_{BR}\right) \Biggr\|_1\\
  &\leq\left(\sum_{i=0}^{n-k-1}\alpha_0^i\right)\max_{\rho_{AR}\in\mathsf{D}(AR)}\sum_{\xi\in\Xi}p_\xi\\*
  &\hspace{7em}\frac{1}{2}\left\|
    \left(\mathcal{O}_{A\to B}^\xi-\alpha_0\Psi_{AR\to BR}\right)(\rho_{AR})
  \right\|_1\\
  &=\left(\sum_{i=0}^{n-k-1}\alpha_0^i\right)\max_{\Ket{\phi}_{AR}\in\mathsf{S}(AR)}\sum_{\xi\in\Xi}p_\xi\\*
  &\hspace{7em}\frac{1}{2}\left\|
    \left(\mathcal{O}_{A\to B}^\xi-\alpha_0\Psi_{AR\to BR}\right)(\ketbra{\phi}{\phi}_{AR})
  \right\|_1\\
  &=\left(\sum_{i=0}^{n-k-1}\alpha_0^i\right)\theta_\diamond^0(\Psi_{AR\to BR}).
\end{align*}
The first inequality and the second equality follows from 
the convexity of the trace distance.
The second inequality follows from the monotonicity of the trace norm.

Similarly, we obtain
\begin{align*}
  &\sum_{j=0}^{k-m-1}\frac{1}{2}\left\|
    \alpha_1^{j+1}\sigma_W^{m+j+1}-\alpha_1^j\sigma_W^{m+j}
  \right\|_1\\
  &\leq\left(\sum_{j=0}^{k-m-1}\alpha_1^j\right)\theta_\diamond^1(\Psi_{AR\to BR}).
\end{align*}

For the last term, we obtain
\begin{align*}
  &\frac{1}{2}\left\|\sigma_W^m-\Ketbra{0}{0}_W\right\|_1\\
  &=\sum_{\xi\in\Xi}p_\xi\sum_{\eta\colon\xi\in C_\eta}
    \Tr\left(M^\eta_{BR}\rho_{BR}^{n-m,n,\xi}\right)\\
  &\leq 1-\perr(m).
\end{align*}

From the above inequalities, we obtain
\begin{align*}
  1
  &\leq\perr(n)+\left(\sum_{i=0}^{n-k-1}\alpha_0^i\right)\theta_\diamond^0(\Psi_{AR\to BR})\\
  &+\left(\sum_{i=0}^{k-m-1}\alpha_1^i\right)\theta_\diamond^1(\Psi_{AR\to BR})+1-\perr(m).
\end{align*}
By replacing last $n-k$ applications of $\Psi_{AR\to BR}$ with $\Psi_{AR\to BR}^0$ and first $k$ applications of $\Psi_{AR\to BR}$ with $\Psi_{AR\to BR}^1$,
we obtain Theorem~\ref{boundForQCGDPWithWeightedTraceDistance}.

\section{Optimization by semidefinite programming}\label{sec:sdp}
In this section, we assume $\dim(R)\geq\dim(A)$, and show that
all optimization problems in Theorems~\ref{boundForQCDPWithBuresAngle},
\ref{boundForQCDPWithWeightedTraceDistance},~\ref{boundForQCGDPWithBuresAngle},
and~\ref{boundForQCGDPWithWeightedTraceDistance} can be formulated as SDPs, 
so that the evaluation will be efficient.
Let $J(\Phi_{A\to B})\in\mathsf{L}(BA)$ be the Choi representations of 
a quantum channel $\Phi_{A\to B}$, defined by
\begin{equation*}
  J(\Phi_{A\to B})
  \coloneqq\sum_{i,j}\Phi_{A\to B}(\ketbra{i}{j}_A)\otimes\ketbra{i}{j}_A.
\end{equation*}
The minimization of the fidelity
$F\bigl(\mathcal{O}_{A\to B}^0(\ketbra{\phi}{\phi}_{AR}),\allowbreak\mathcal{O}_{A\to B}^1(\ketbra{\phi}{\phi}_{AR})\bigr)$
in Theorem~\ref{boundForQCDPWithBuresAngle}
can be formalized as the following SDP~\cite{katariya2021geometric}:
\begin{align*}
  \max_{\lambda,X_{BA}}\colon& \lambda\\
  \text{subject to}\colon&
  \begin{bmatrix}
    J(\mathcal{O}_{A\to B}^0) & X_{BA} \\
    X_{BA}^\dagger & J(\mathcal{O}_{A\to B}^1)
  \end{bmatrix}
  \succeq 0\\
  &\frac12\left(\Tr_B(X_{BA})+\Tr_B(X_{BA}^\dagger)\right)\succeq\lambda I_A.
\end{align*}
To make the present paper self-contained, a simple proof is shown in Appendix~\ref{apx:sdpf}.
Similarly, for Theorem~\ref{boundForQCGDPWithBuresAngle},
\begin{align*}
  &\max_{\Psi_{A\to B}}\theta_{\mathcal{A}}(\Psi_{A\to B}\otimes\mathrm{id}_R)=\max_{\Psi_{A\to B}}\min_{\ket{\phi}_{AR}\in\mathsf{S}(AR)}\sum_{\xi\in\Xi}p_\xi\\
  &\hspace{3em}\F{\mathcal{O}_{A\to B}^\xi(\ketbra{\phi}{\phi}_{AR})}{\Psi_{A\to B}(\ketbra{\phi}{\phi}_{AR})}
\end{align*}
can be formalized as the following SDP:
\begin{align*}
  \max_{\lambda,X_{BA}^\xi,P_{BA}}\colon& \lambda\\
  \text{subject to}\colon&
  \begin{bmatrix}
    J(\mathcal{O}_{A\to B}^\xi) & X_{BA}^\xi \\
    X_{BA}^{\xi\dagger} & P_{BA}
  \end{bmatrix}
  \succeq 0 \quad (\forall\xi\in\Xi)\\
  &\frac12\sum_{xi\in\Xi}p_\xi\left(\Tr_B(X_{BA}^\xi)+\Tr_B(X_{BA}^{\xi\dagger})\right)\succeq\lambda I_A\\
  &\Tr_B(P_{BA})=I_A.
\end{align*}
Here, $P_{BA}$ corresponds to the Choi operator $J(\Psi_{A\to B})$.

For Theorems~\ref{boundForQCDPWithWeightedTraceDistance} 
and~\ref{boundForQCGDPWithWeightedTraceDistance}, 
we use Watrous's SDP using the Choi operator~\cite{watrous2013simpler}.
Indeed, $\|\mathcal{O}_{A\to B}^0-\alpha\mathcal{O}_{A\to B}^1\|_\diamond$
in Theorem~\ref{boundForQCDPWithWeightedTraceDistance} is a solution 
of the following SDP~\cite{watrous2013simpler}:
\begin{align*}
  \min_{\lambda,Y_{BA},Z_{BA}}\colon& \lambda\\
  \text{subject to}\colon&
  J(\mathcal{O}_{A\to B}^0)-\alpha J(\mathcal{O}_{A\to B}^1)=Y_{BA}-Z_{BA},\\
  &\lambda I_A\succeq\Tr_B(Y_{BA})+\Tr_B(Z_{BA}).
\end{align*}
The proof is shown in Appendix~\ref{apx:sdpt}.
For Theorem~\ref{boundForQCGDPWithWeightedTraceDistance},
\begin{align*}
  &\min_{\Psi_{A\to B}^0}\theta_\diamond^0(\Psi_{A\to B}^0\otimes\mathrm{id}_R)=\min_{\Psi_{A\to B}^0}\max_{\ket{\phi}_{AR}}\sum_{\xi\in\Xi}p_\xi\\
  &\hspace{3em}\frac{1}{2}\left\|
    \left(\mathcal{O}_{A\to B}^\xi-\alpha_0\Psi_{A\to B}^0\right)(\ketbra{\phi}{\phi}_{AR})
  \right\|_1
\end{align*}
can be represented as a solution of SDP\@.
Similarly to the above case, $\min_{\Psi_{A\to B}^0}\theta_\diamond^0(\Psi_{A\to B}^0\otimes\mathrm{id}_R)$
is a solution of the following SDP:
  
\begin{align*}
  \min_{\lambda,Y_{BA}^\xi,Z_{BA}^\xi,P_{BA}}\colon& \lambda\\
  \text{subject to}\colon&
  J(\mathcal{O}_{A\to B}^\xi)-\alpha_0 P_{BA} = Y_{BA}^\xi - Z_{BA}^\xi\\
  & \hspace{14em} (\forall\xi\in\Xi)\\
  &\lambda I_A\succeq \frac12\sum_{\xi\in\Xi}p_\xi\left(\Tr_B(Y_{BA}^\xi)+\Tr_B(Z_{BA}^\xi)\right)\\
  &P_{BA} \succeq 0,\quad \Tr_B(P_{BA})=I_A.
\end{align*}

Hence, all lower bounds obtained in this paper can be evaluated and optimized on the condition $\Psi_{AR\to BR}=\Psi_{A\to B}\otimes\mathrm{id}_R$ efficiently.

\section{Concluding remarks}\label{sec:concl}
In this study, we derive lower bounds on the error probability of quantum channel group discrimination problem on the basis of the triangle inequalities for the 
Bures angle and the trace distance.
It would be interesting to investigate if other distance measures exist that yield tighter and efficiently computable lower bounds.

The function $\cos\theta\mapsto\cos n\theta$ is known as the Chebyshev polynomial.
The Chebyshev polynomial has some extremal properties.
The polynomial method for lower bounding the query complexity is relevant to the Chebyshev polynomial~\cite{beals2001quantum,paturi1992degree}.
Hence, it is interesting to study the connection between the lower bounds in Theorems~\ref{boundForQCDPWithBuresAngle} and~\ref{boundForQCGDPWithBuresAngle} using the Bures angle
and the polynomial method.

Our technique can be used together with the technique based on the port-based teleportation~\cite{pirandola2019fundamental}.
The combination of two techniques may improve the lower bound.

\begin{acknowledgments}
This work was supported by JST PRESTO Grant Number JPMJPR1867, JST FOREST Grant Number JPMJFR216V.
and JSPS KAKENHI Grant Numbers JP17K17711, JP18H04090, JP20H04138, and JP20H05966.
We would like to thank Editage (\url{www.editage.com}) for English language editing.
\end{acknowledgments}

\appendix
\section{Triangle inequalities for the Bures distance and the sine distance}\label{apx:triangle}
In this appendix, we show the triangle inequalities for the Bures distance and the sine distance from the triangle inequality for the Bures angle.
For arbitrary quantum states $\rho_A,\,\sigma_A,\,$ and $\tau_A\in\mathsf{D}(A)$, the Bures angle satisfies the triangle inequality
\begin{equation*}
  \mathcal{A}(\rho_A, \sigma_A) 
  \leq\mathcal{A}(\rho_A,\tau_A)+\mathcal{A}(\tau_A, \sigma_A).
\end{equation*}
Then,
\begin{align*}
  \mathcal{B}(\rho_A, \sigma_A) 
  &=2\sin\left(\frac{\mathcal{A}(\rho_A,\sigma_A)}2\right)\\
  &\leq 2\sin\left(\frac{\mathcal{A}(\rho_A,\tau_A)+\mathcal{A}(\tau_A,\sigma_A)}{2}\right)
\end{align*}
because the sine function is monotonically increasing in $[0,\pi/2]$.
From the trigonometric addition formula,
\begin{align*}
  &2\sin\left(\frac{\mathcal{A}(\rho_A,\tau_A)+\mathcal{A}(\tau_A,\sigma_A)}{2}\right)\\
  &=2\sin\left(\frac{\mathcal{A}(\rho_A,\tau_A)}{2}\right)
  \cos\left(\frac{\mathcal{A}(\tau_A,\sigma_A)}{2}\right)\\
  &\quad+2\sin\left(\frac{\mathcal{A}(\tau_A,\sigma_A)}{2}\right)
  \cos\left(\frac{\mathcal{A}(\rho_A,\tau_A)}{2}\right)\\
  &\leq 2\sin\left(\frac{\mathcal{A}(\rho_A,\tau_A)}{2}\right)
  +2\sin\left(\frac{\mathcal{A}(\tau_A,\sigma_A)}{2}\right).
\end{align*}
This proves the triangle inequality for the Bures distance.

The triangle inequality for the sine distance
\begin{align*}
  \sin\left(\mathcal{A}(\rho_A,\sigma_A)\right)
  \leq\sin\left(\mathcal{A}(\rho_A,\tau_A)\right)
  +\sin\left(\mathcal{A}(\tau_A,\sigma_A)\right)
\end{align*}
can be obtained in the same way if
$\mathcal{A}(\rho_A,\tau_A)+\mathcal{A}(\tau_A,\sigma_A)\leq\frac{\pi}{2}$.
If $\mathcal{A}(\rho_A,\tau_A)+\mathcal{A}(\tau_A,\sigma_A)>\frac{\pi}{2}$,
the triangle inequality is obtained from
\begin{align*}
  &\sin\left(\mathcal{A}(\rho_A,\tau_A)\right)+\sin\left(\mathcal{A}(\tau_A,\sigma_A)\right)\\
  &\ge\frac{2}{\pi}\left(\mathcal{A}(\rho_A,\tau_A)+\mathcal{A}(\tau_A,\sigma_A)\right) > 1.
\end{align*}

\section{Proof of Lemma~\ref{fuchsVanDeGraaf}}\label{apx:fuchsVanDeGraaf}
The following proof is a generalization of the proof of
Fuchs--van~de~Graaf inequality in~\cite{watrous2018theory}.
From Uhlmann's theorem, there exist pure states
$\Ket{\phi^0}_{AR},\Ket{\phi^1}_{AR}\in\mathsf{S}(AR)$ satisfying the equalities
$\Tr_R\left(\Ketbra{\phi^0}{\phi^0}_{AR}\right)=\rho_A^0$,
$\Tr_R\left(\Ketbra{\phi^1}{\phi^1}_{AR}\right)=\rho_A^1$, and
$\left|\Braket{\phi^0|\phi^1}_{AR}\right|=F(\rho_A^0,\rho_A^1)$.
Therefore, we obtain
\begin{align*}
  &\left\|a_0\rho_A^0-a_1\rho_A^1\right\|_1\\
  &=\left\|
    a_0\Tr_R\left(\Ketbra{\phi^0}{\phi^0}_{AR}\right)
    -a_1\Tr_R\left(\Ketbra{\phi^1}{\phi^1}_{AR}\right)
  \right\|_1\\
  &\leq\left\|
    a_0\Ketbra{\phi^0}{\phi^0}_{AR}-a_1\Ketbra{\phi^1}{\phi^1}_{AR}
  \right\|_1\\
  &=\sqrt{(a_0+a_1)^2-4a_0a_1\left|\Braket{\phi^0|\phi^1}_{AR}\right|^2}\\
  &=\sqrt{(a_0+a_1)^2-4a_0a_1\F{\rho_A^0}{\rho_A^1}^2}.
\end{align*}
The inequality follows from the monotonicity of the trace norm.

\section{Minimum of the fidelity}\label{apx:ft}
\begin{lemma}\label{lem:ft}
  Let $\mathcal{O}_{A\to B}^0$ and $\mathcal{O}_{A\to B}^1$ be quantum channels with 
  Stinespring representations $O_{A\to BE}^0$ and $O_{A\to BE}^1$, respectively.
  If $\dim(R)\geq\dim(A)$, the following holds:
  \begin{align*}
    &\min_{\ket{\phi}_{AR}\in\mathsf{S}(AR)}
      \F{\mathcal{O}_{A\to B}^0(\ketbra{\phi}{\phi}_{AR})}{\mathcal{O}_{A\to B}^1(\ketbra{\phi}{\phi}_{AR})}\\
    &=\min_{\sigma_A\in\mathsf{D}(A)}
      \left\|\Tr_B\left(O_{A\to BE}^0\sigma_{A}O_{A\to BE}^{1\dagger}\right)\right\|_1.
  \end{align*}
\end{lemma}
\begin{proof}
  \begin{align*}
    &\F{\Tr_E\left(\Ketbra{\psi^0}{\psi^0}_{CE}\right)}{\Tr_E\left(\Ketbra{\psi^1}{\psi^1}_{CE}\right)}\\
    &=\max_{U_E}\left|\Braket{\psi^1|U_E|\psi^0}_{CE}\right|\\
    &=\max_{U_E}\left|\Tr_{CE}\left(U_E\Ketbra{\psi^0}{\psi^1}_{CE}\right)\right|\\
    &=\max_{U_E}\left|\Tr_E\left(U_E\Tr_C\left(\Ketbra{\psi^0}{\psi^1}_{CE}\right)\right)\right|\\
    &=\left\|\Tr_C\left(\Ketbra{\psi^0}{\psi^1}_{CE}\right)\right\|_1,
  \end{align*}
  where $U_E$ is a unitary operator.
  The first equality follows from Uhlmann's theorem.
  The last equality follows from the duality of the Schatten $p$-norm.

  We then obtain
  \begin{align*}
    &\F{\mathcal{O}_{A\to B}^0(\ketbra{\phi}{\phi}_{AR})}{\mathcal{O}_{A\to B}^1(\ketbra{\phi}{\phi}_{AR})}\\
    &=F\biggl(\Tr_{E}\left(O_{A\to BE}^0\ketbra{\phi}{\phi}_{AR}O_{A\to BE}^{0\dagger}\right),\\
    &\qquad \Tr_{E}\left(O_{A\to BE}^1\ketbra{\phi}{\phi}_{AR}O_{A\to BE}^{1\dagger}\right)\biggr)\\
    &=\left\|\Tr_{BR}\left(O_{A\to BE}^0\ketbra{\phi}{\phi}_{AR}O_{A\to BE}^{1\dagger}\right)\right\|_1\\
    &=\left\|\Tr_{B}\left(O_{A\to BE}^0\Tr_R\left(\ketbra{\phi}{\phi}_{AR}\right)O_{A\to BE}^{1\dagger}\right)\right\|_1.
  \end{align*}
  When $\dim(R)\geq\dim(A)$, any $\sigma_A\in\mathsf{D}(A)$ can be represented by $\Tr_R(\ketbra{\phi}{\phi}_{AR})$ for some $\ket{\phi}_{AR}\in\mathsf{S}(AR)$.
  Hence, Lemma is proved.
\end{proof}

\section{SDP representations}
\subsection{SDP for the fidelity}\label{apx:sdpf}
The fidelity of $\rho_A$ and $\sigma_A$ is represented as an solution of the following SDP~\cite{watrous2018theory}.
\begin{align*}
\min_{Y_Z,Z_A}:\,& \frac12\left(\langle\rho_A, Y_A\rangle + \langle\sigma_A, Z_A\rangle\right)\\
\text{subject to}:\,&
\begin{bmatrix}
Y_A&I_A\\
I_A&Z_A
\end{bmatrix}\succeq 0.
\end{align*}
Hence, $\min_{\ket{\psi}_{AR}}F(\Phi_{A\to B}(\ketbra{\psi}{\psi}_{AR}), \Psi_{A\to B}(\ketbra{\psi}{\psi}_{AR}))$ is the optimal value of
\begin{equation}\label{eq:df}
\begin{split}
\min_{\ket{\psi}_{AR},Y_{BR},Z_{BR}}:& \frac12\biggl(\bigl\langle\Phi_{A\to B}(\ketbra{\psi}{\psi}_{AR}), Y_{BR}\bigr\rangle\\
&\qquad + \bigl\langle\Psi_{A\to B}(\ketbra{\psi}{\psi}_{AR}), Z_{BR}\bigr\rangle\biggr)\\
\text{subject to}:&
\begin{bmatrix}
Y_{BR}&I_{BR}\\
I_{BR}&Z_{BR}
\end{bmatrix}\succeq 0
\end{split}
\end{equation}
For $\ket{\psi}_{AR}=\sum_{a,c}\psi_{a,c}\ket{a}_A\ket{c}_R$, we obtain 
\begin{align}
&\Phi_{A\to B}(\ketbra{\psi}{\psi}_{AR})=\sum_{a,b,c,d}\psi_{a,c}\psi^*_{b,d}\Phi_{A\to B}(\ketbra{a}{b}_A)\otimes\ketbra{c}{d}_R\nonumber\\
&\quad =\sum_{a,b}\Phi_{A\to B}(\ketbra{a}{b}_A)\otimes\sum_{c,d}\psi_{a,c}\psi_{b,d}^*\ketbra{c}{d}_R\nonumber\\
&\quad =\sum_{a,b}\Phi_{A\to B}(\ketbra{a}{b}_A)\otimes L_{A\to R}\ketbra{a}{b}_A L_{A\to R}^\dagger\nonumber\\
&\quad =(I_B\otimes L_{A\to R})J(\Phi_{A\to B})(I_B\otimes L_{A\to R}^\dagger)\label{eq:cp}
\end{align}
where $L_{A\to R}:=\sum_{a,c}\psi_{a,c}\ket{c}_R\bra{a}_A$.
Let
\begin{align*}
Y_{BA}' &:=(I_B\otimes L_{A\to R}^\dagger)Y_{BR}(I_B\otimes L_{A\to R})\\
Z_{BA}' &:=(I_B\otimes L_{A\to R}^\dagger)Z_{BR}(I_B\otimes L_{A\to R}).
\end{align*}
Then, we obtain
\begin{align*}
\bigl\langle\Phi_{A\to B}(\ketbra{\psi}{\psi}_{AR}), Y_{BR}\bigr\rangle &= \bigl\langle J(\Phi_{A\to B}), Y'_{BA}\bigr\rangle\\
\bigl\langle\Psi_{A\to B}(\ketbra{\psi}{\psi}_{AR}), Z_{BR}\bigr\rangle &= \bigl\langle J(\Phi_{A\to B}), Z'_{BA}\bigr\rangle
\end{align*}
and
\begin{align*}
\begin{bmatrix}
Y_{BR}&I_{BR}\\
I_{BR}&Z_{BR}
\end{bmatrix}\succeq 0
\Longrightarrow
L
\begin{bmatrix}
Y_{BR}&I_{BR}\\
I_{BR}&Z_{BR}
\end{bmatrix}L^\dagger\succeq 0
\end{align*}
where
\begin{align*}
L := 
\begin{bmatrix}
I_B\otimes L_{A\to R}^\dagger&0\\
0&I_B \otimes L_{A\to R}^\dagger
\end{bmatrix}.
\end{align*}
Since
\begin{align*}
L
\begin{bmatrix}
Y_{BR}&I_{BR}\\
I_{BR}&Z_{BR}
\end{bmatrix}L^\dagger
=
\begin{bmatrix}
Y'_{BA}&I_{B}\otimes\sigma_A\\
I_{B}\otimes\sigma_A&Z'_{BA}
\end{bmatrix}
\end{align*}
where $\sigma_A:=L_{A\to R}^\dagger L_{A\to R}=\sum_{a,b}(\sum_c\psi_{a,c}^*\psi_{b,c})\ketbra{a}{b}_A$,
the optimal value of the following SDP is at most the optimal value of~\eqref{eq:df}.
\begin{align}
\min_{Y'_{BA},\,Z'_{BA},\,\sigma_A}:& \frac12\left(\bigl\langle J(\Phi_{A\to B}), Y'_{BA}\bigr\rangle + \bigl\langle J(\Psi_{A\to B}), Z'_{BA}\bigr\rangle\right)\nonumber\\
\text{subject to}:&
\begin{bmatrix}
Y'_{BA}&I_{B}\otimes\sigma_A\\
I_{B}\otimes\sigma_A&Z'_{BA}
\end{bmatrix}\succeq 0\label{eq:sdpf}\\
&\sigma_A\succeq 0, \Tr(\sigma_A)=1.\nonumber
\end{align}

Next, we will show that the optimal value of~\eqref{eq:sdpf} is equal to the optimal value of~\eqref{eq:df} if $\dim(R)\ge\dim(A)$.
Even if we replace $\sigma_A\succeq 0$ with $\sigma_A\succ 0$, and replace $\min$ with $\inf$ in~\eqref{eq:sdpf}, the optimal value of~\eqref{eq:sdpf} is unchanged from the continuity.
Let $(Y'_{BA}, Z'_{BA}, \sigma_A)$ be an arbitrary feasible solution of~\eqref{eq:sdpf} with $\sigma_A\succ 0$.
Let $\sigma_A=\sum_{a}\lambda_a\ketbra{\varphi_a}{\varphi_a}_A$ be the spectral decomposition of $\sigma_A$.
From $\dim(R)\ge\dim(A)$, $\sigma_A=L_{A\to R}^\dagger L_{A\to R}$ for $L_{A\to R}:=\sum_{a}\sqrt{\lambda_a}\ket{a}_R\bra{\varphi_a}_A$.
Let $J_{A\to R}:=\sum_{a}\frac1{\sqrt{\lambda_a}}\ket{a}_R\bra{\varphi_a}_A$. Then, $L_{A\to R}^\dagger J_{A\to R}=I_A$.
Then, it is easy to verify that
\begin{align*}
Y_{BR}:= (I_B\otimes J_{A\to R}) Y'_{BA}(I_B\otimes J_{A\to R}^\dagger)\\
Z_{BR}:= (I_B\otimes J_{A\to R}) Z'_{BA}(I_B\otimes J_{A\to R}^\dagger)\\
\ket{\psi} := \sum_{a}\sqrt{\lambda_a}\ket{\varphi_a}_A^*\ket{a}_R
\end{align*}
is a feasible solution of~\eqref{eq:df}, and have the same objective value.
This implies that the optimal value of~\eqref{eq:sdpf} is equal to the optimal value of~\eqref{eq:df} if $\dim(R)\ge\dim(A)$.

The dual of~\eqref{eq:sdpf} is 
\begin{align*}
&\max_{\lambda,\begin{bmatrix}U_{BA}&W_{BA}\\ W_{BA}^\dagger&V_{BA}\end{bmatrix}\succeq0}\min_{Y'_{BA},\,Z'_{BA},\,\sigma_A\succeq0}\\
&\quad\frac12\left(\bigl\langle J(\Phi_{A\to B}), Y'_{BA}\bigr\rangle + \bigl\langle J(\Psi_{A\to B}), Z'_{BA}\bigr\rangle\right)\\
&\quad-\frac12\left\langle
\begin{bmatrix}
U_{BA}&W_{BA}\\
W_{BA}^\dagger&V_{BA}
\end{bmatrix},
\begin{bmatrix}
Y'_{BA}&I_{B}\otimes\sigma_A\\
I_{B}\otimes\sigma_A&Z'_{BA}
\end{bmatrix}
\right\rangle\\
&\quad -\lambda(\Tr(\sigma_A)-1)\\
&= \max\min\\
&\quad\frac12\left(\bigl\langle J(\Phi_{A\to B}), Y'_{BA}\bigr\rangle + \bigl\langle J(\Psi_{A\to B}), Z'_{BA}\bigr\rangle\right)\\
&\quad-\frac12\langle U_{BA}, Y_{BA}'\rangle - \frac12\langle W_{BA}+W_{BA}^\dagger, I_B\otimes \sigma_A\rangle\\
&\quad - \frac12\langle V_{BA}, Z_{BA}'\rangle-\lambda(\Tr(\sigma_A)-1)\\
&= \max\min\\
&\quad\frac12\Bigl(\bigl\langle J(\Phi_{A\to B}) - U_{BA}, Y'_{BA}\bigr\rangle\\
&\qquad + \bigl\langle J(\Psi_{A\to B})-V_{BA}, Z'_{BA}\bigr\rangle\Bigr)\\
&\quad-\left\langle \lambda I_A + \frac12(W_{A}+W_{A}^\dagger), \sigma_A\right\rangle +\lambda.
\end{align*}
Hence, the dual SDP of~\eqref{eq:sdpf} is
\begin{align*}
\max:\,&\lambda\\
\text{subject to}:\,&
\begin{bmatrix}
J(\Phi_{A\to B})& W_{BA}\\
W_{BA}^\dagger& J(\Psi_{A\to B})
\end{bmatrix}\succeq 0\\
&\lambda I_A + \frac12(W_A+W_A^\dagger)\preceq 0.
\end{align*}
Since
\begin{align*}
&\begin{bmatrix}
J(\Phi_{A\to B})& W_{BA}\\
W_{BA}^\dagger& J(\Psi_{A\to B})
\end{bmatrix}\succeq 0\\
&\iff
\begin{bmatrix}
J(\Phi_{A\to B})& -W_{BA}\\
-W_{BA}^\dagger& J(\Psi_{A\to B})
\end{bmatrix}
\succeq 0
\end{align*}
we obtain the SDP
\begin{align*}
\max:\,&\lambda\\
\text{subject to}:\,&
\begin{bmatrix}
J(\Phi_{A\to B})& X_{BA}\\
X_{BA}^\dagger& J(\Psi_{A\to B})
\end{bmatrix}\succeq 0\\
&\frac12(X_A+X_A^\dagger)\succeq \lambda I_A.
\end{align*}

\subsection{SDP for the trace norm}\label{apx:sdpt}
The trace norm of an Hermitian matrix $H_A$ is represented as an solution of the following SDP\@.
\begin{align*}
\max_{X_A}:\,& \langle H_A, X_A\rangle\\
\text{subject to}:\,&
-I_A\preceq X_A \preceq I_A.
\end{align*}
Hence, for an Hermitian-preserving map $\Lambda_{A\to B}$, $\max_{\ket{\psi}_{AR}} \|\Lambda_{A\to B}(\ketbra{\psi}{\psi}_{AR})\|_1$ is equal to
\begin{equation}
\begin{split}
\max_{\ket{\psi}_{AR},\, X_{BR}}:\,& \left\langle \Lambda_{A\to B}(\ketbra{\psi}{\psi}_{AR}), X_{BR}\right\rangle\\
\text{subject to}:\,&
-I_{BR}\preceq X_{BR} \preceq I_{BR}.
\end{split}
\label{eq:dia}
\end{equation}
From~\eqref{eq:cp},
\begin{equation*}
\Lambda_{A\to B}(\ketbra{\psi}{\psi}_{AR}) =(I_B\otimes L_{A\to R})J(\Lambda_{A\to B})(I_B\otimes L_{A\to R}^\dagger).
\end{equation*}
Let
\begin{equation*}
X'_{BA} := (I_B\otimes L_{A\to R}^\dagger) X_{BR}(I_B\otimes L_{A\to R}).
\end{equation*}
Then, we obtain
\begin{equation*}
\left\langle\Lambda_{A\to B}(\ketbra{\psi}{\psi}_{AR}), X_{BR}\right\rangle
=
\left\langle J(\Lambda_{A\to B}), X_{BA}'\right\rangle
\end{equation*}
and
\begin{equation*}
-I_{BR}\preceq X_{BR} \preceq I_{BR}\Longrightarrow
-I_{B}\otimes\sigma_A\preceq X'_{BA} \preceq I_{B}\otimes\sigma_A.
\end{equation*}
Hence, the optimal value of the following SDP is at least the optimal value of~\eqref{eq:dia}.
\begin{equation}
\begin{split}
\max:\,& \left\langle J(\Lambda_{A\to B}), X_{BA}'\right\rangle\\
\text{subject to}:\,&
-I_{B}\otimes\sigma_A\preceq X_{BA}'\preceq I_B\otimes \sigma_A\\
&\sigma_A\succeq0, \Tr(\sigma_A)=1.
\end{split}\label{eq:sdpdia}
\end{equation}
When $\dim(R)\ge\dim(A)$, we can show that the optimal value of~\eqref{eq:sdpdia} is equal to the optimal value of~\eqref{eq:dia} in the same way as the previous section.

The dual of~\eqref{eq:sdpdia} is
\begin{align*}
&\min_{\lambda,Y_{BA}\succeq0,Z_{BA}\succeq0} \,\max_{X'_{BA},\sigma_A\succeq0} \langle J(\Lambda_{A\to B}), X_{BA}'\rangle - \lambda(\Tr(\sigma_A)-1)\\
&\qquad + \langle Y_{BA}, I_B\otimes\sigma_A - X'_{BA}\rangle + \langle Z_{BA}, I_B\otimes\sigma_A+X_{BA}'\rangle\\
&=\min_{\lambda, Y_{BA}\succeq0,Z_{BA}\succeq0}\max_{X'_{BA},\sigma_A\succeq0}\\
&\qquad \lambda + \langle J(\Lambda_{A\to B}) -Y_{BA}+Z_{BA}, X_{BA}'\rangle\\
&\qquad + \langle Y_{A}+Z_{A}-\lambda I_A,\sigma_A\rangle.
\end{align*}
Hence, the dual SDP of~\eqref{eq:sdpdia} is
\begin{align*}
\min_{\lambda, Y_{BA},Z_{BA}}:\,&\lambda\\
\text{subject to}:\,&  J(\Lambda_{A\to B})= Y_{BA}-Z_{BA}\\
&  Y_A+Z_A\preceq \lambda I_A\\
& Y_{BA}\succeq 0,\,Z_{BA}\succeq0.
\end{align*}
Note that $\|\Tr_B(|J(\Lambda_{A\to B})|)\|_\infty$ is an upper bound of $\|\Lambda_{A\to B}\|_\diamond$, but not neccesarily equal to $\|\Lambda_{A\to B}\|_\infty$.

\bibliography{main}

\end{document}